\def\year{2020}\relax
\documentclass[letterpaper]{article} % DO NOT CHANGE THIS
\usepackage{aaai20}  % DO NOT CHANGE THIS
\usepackage{times}  % DO NOT CHANGE THIS
\usepackage{helvet} % DO NOT CHANGE THIS
\usepackage{courier}  % DO NOT CHANGE THIS
\usepackage[hyphens]{url}  % DO NOT CHANGE THIS
\usepackage{graphicx} % DO NOT CHANGE THIS
\usepackage{graphicx}  % DO NOT CHANGE THIS
\usepackage{color}

\usepackage{amsthm}
\usepackage{amssymb}
\usepackage{amsmath}
\usepackage{booktabs}
\usepackage[algo2e,ruled]{algorithm2e}
\usepackage[margin=10pt]{subfig}
\usepackage{multirow}
\usepackage{siunitx}

\theoremstyle{definition}

\newtheorem{theorem}{Theorem}[section]
\newtheorem*{theorem*}{Theorem}

\newtheorem{definition}[theorem]{Definition}
\newtheorem{lemma}[theorem]{Lemma}
\newtheorem*{lemma_nonum}{Lemma}

\theoremstyle{remark}

\newtheorem{remark}{Remark}[section]

\newcommand\blfootnote[1]{%
  \begingroup
  \renewcommand\thefootnote{}\footnote{#1}%
  \addtocounter{footnote}{-1}%
  \endgroup
}

\nocopyright
\title{Active learning in the geometric block model }
\author{\textbf{Eli Chien,\textsuperscript{\rm 1}\thanks{This work is done during Eli Chien's internship at Nokia Bell Labs.}}\textbf{Antonia Maria Tulino,\textsuperscript{\rm 2,3}}\textbf{Jaime Llorca\textsuperscript{\rm 2}}\\
\textsuperscript{\rm 1}ECE, University of Illinois Urbana-Champaign \\
\textsuperscript{\rm 2}Nokia Bell Labs \textsuperscript{\rm 3}DIETI, University of Naples Federico II\\
ichien3@illinois.edu, \{a.tulino, jaime.llorca\}@nokia.com
}
 
\begin{document}

\maketitle

\begin{abstract}
The geometric block model is a recently proposed generative model for random graphs that is able to capture the inherent geometric properties of many community detection problems, providing more accurate characterizations of practical community structures compared with the popular stochastic block model.
Galhotra et al. recently proposed a motif-counting algorithm for unsupervised community detection in the geometric block model that is proved to be near-optimal. They also characterized the regimes of the model parameters for which the proposed algorithm can achieve exact recovery.
In this work,
we initiate the study of active learning in the geometric block model.
That is,
we are interested in the problem of exactly recovering the community structure of random graphs following the geometric block model under arbitrary model parameters,
by possibly querying the labels of a limited number of chosen nodes.
We propose two active learning algorithms that combine the idea of motif-counting with two different label query policies. 
Our main contribution is to show that sampling the labels of a vanishingly small fraction of nodes (sub-linear in the total number of nodes) is sufficient to achieve exact recovery 
in the regimes under which the state-of-the-art unsupervised method fails.
We validate the superior performance of our algorithms via numerical simulations on both real and synthetic datasets. \blfootnote{The conference version of this paper will appear in AAAI-20.}
\end{abstract}

\section{Introduction}
%graphs have  become extremely useful  representations of a wide variety of systems in different areas. Biological, social, technological, and information networks can be studied as graphs, and graph analysis has become crucial to understand the features of these systems
%Communities, also called clusters or modules, are groups of vertices which probably share common properties and/or play similar roles within the graph. In Fig. 1 a schematic example of a graph with communities is shown.

Community detection (or graph clustering) is one of the most important tasks in machine learning and data mining. In this problem, it is assumed that each node (or vertex) in a network (or graph) belongs to one of the underlying communities (or clusters), and that the topology of the network depends on these latent group memberships (or labels).
The goal is to recover the communities by partitioning the nodes into different classes that match the labels up to a permutation. This problem has many applications, such as clustering in social networks~\cite{fortunato2010community}, detecting protein complexes in protein interaction networks~\cite{chen2006detecting}, identifying customer interests in recommendation systems~\cite{pittir13328}, and performing image classification and segmentation~\cite{shi2000normalized}.

The stochastic block model (SBM) is a popular random graph model for community detection that generalizes the well-known Erd{\"o}s-Renyi model~\cite{holland1983stochastic,mossel2015consistency}.
In the SBM, the probability of having an edge between a pair of nodes depends only on the labels of the corresponding two nodes. In its simplest version, the SBM contains two communities of equal sizes, such that a pair of nodes from the same community are connected with probability $p$, and nodes from different communities are connected with probability $q$.
Prior works (see~\cite{abbe2017community} for an overview), have established the limits of unsupervised methods to achieve exact community detection in terms of the relative difference between $p$ and $q$.
%{\RED The limits on the performance of unsupervised methods in terms of the relative difference between $p$ and $q$ have been established in prior works (see~\cite{abbe2017community} for an overview).}

I has then become an important question in practice to understand if we can still recover the correct community memberships in the regimes where unsupervised methods fail, by querying the labels of a small subset of nodes.
The process of actively querying the labels of a subset of nodes, referred to as {\em active learning}, is a very useful tool for many machine learning applications where the acquisition of labeled data is expensive and/or time consuming~\cite{cohn1994improving}.
In the active learning framework, we are allowed to query node labels up to a budget constraint in order to improve overall clustering accuracy. % while guaranteeing a desired level of clustering performance.
The authors of~\cite{gadde2016active} % give an affirmative answer to this question for the case of the SBM.
showed that a sub-linear number of queries is sufficient to achieve exact recovery below the limit (in terms of difference between $p$ and $q$) for unsupervised methods in the SBM,
and that the number of queries needed for exact recovery depends on how far we are below the limit -- hence providing a smooth trade-off between query complexity and clustering hardness in the SBM.

 %{\RED A natural question thus arises: Can we still recover the communities by querying a small amount of labels below the limits? The authors of~\cite{gadde2016active} give an affirmative answer to this question for the case of SBM. They show that using the idea of active learning, we are able to exactly recover the communities in SBM with sub-linear number of queries when we are below the limit for unsupervised methods. Active learning is a very useful idea for many machine learning applications where the acquisition of labeled data is expensive and time consuming~\cite{cohn1994improving}. In the active learning framework, we are allowed to query node labels up to a budget constraint while guaranteeing a desired level of clustering performance. Interestingly, the number of query we need for exact recovery depends on how far we are below the limits. Thus their result provides a smooth trade-off between query complexity and the hardness of clustering for SBM.}

 %{\RED On the other hand,
 While the SBM has gained a lot of popularity to benchmark the performance of clustering algorithms due to its ease of tractability, it fails to capture very important properties of real networks, such as ``transitivity'' (‘friends having common friends’)~\cite{holland1971transitivity,wasserman1994social}.
 %which is an important property in many real world networks especially in social networks~\cite{holland1971transitivity,wasserman1994social}.
 %To be precise, let us
 Consider any three nodes in a graph, $x$, $y$, and $z$. Given the existence of edges between $x$ and $y$, and between $y$ and $z$, (partial) transitivity dictates that it is more likely than not that there also exists an edge between $x$ and $z$.
 However, under the SBM, a generalization of the Erd{\"o}s-Renyi graph model, edges are assumed to be independent of each other, conditioned on their respective node labels. Hence, in the SBM, the existence of edges $(x,y)$ and $(y,z)$ does not affect the probability of having edge $(x,z)$, failing to capture transitivity.

In order to account for the apparent transitivity of many real networks, the authors of~\cite{galhotra2018geometric} proposed a random graph community detection model termed geometric block model (GBM). The GBM combines elements of the SBM with the well studied {\em random geometric graph} (RGG) model that has found important practical applications e.g., in wireless networking~\cite{penrose2003random,gupta1999critical,devroye2011high,goel2005monotone}. In the GBM, the probability that and edge exists between two nodes depends, not only on the associated node labels, but also on their relative distance in the latent feature space. The authors in~\cite{galhotra2018geometric} experimentally validated the benefit of the GBM compared with the SBM to more accurately model real-world networks. % and showed that indeed better models real datasets compared with the SBM in many practical situations.
In their follow up work~\cite{galhotra2018connectivity}, they proposed a state-of-the-art near-optimal motif-counting algorithm that can achieve exact recovery with high probability when the GBM parameters are above the limit of unsupervised methods.
Interestingly, as we illustrate in Section~\ref{limits}, such limit is much higher than in the SBM, showing that clustering in the GBM is fundamentally harder than in the SBM, and hence that in many practical settings, unsupervised methods will not be sufficient to accurately cluster real-world networks.
%in the regime of the sparse graphs. They show that their algorithm is near-optimal will be properly explained later.

\subsection{Contributions}
% {\RED I feel that this is too much of a jump }

Motivated by the advantage of the GBM to characterize of real-world networks and by the increased difficulty in clustering GBM based networks, in this work we initiate the study of active learning in the GBM.
%Surprisingly, we find that the limit of the motif-counting algorithm is much higher then those limits of algorithms for SBM. Therefore, active learning in GBM is even more important then the case of SBM. This motivates our study of active learning on GBM.

We propose two active learning algorithms for the GBM that exactly recover the community memberships with high probability using a sub-linear number of queries, even when we are below the limit of the state-of-the-art unsupervised algorithm in~\cite{galhotra2018connectivity}. Similar to the result of~\cite{gadde2016active} on the SBM, our results offer a smooth trade-off between the query complexity and the hardness of clustering in the GBM.
Both algorithms exploit the idea of motif-counting to remove cross-cluster edges, while combining it with active learning in a different way.
The first algorithm combines motif-counting with the minimax optimal graph-based active learning algorithm $S^2$~\cite{dasarathy2015s2}.
The second algorithm exploits the result on the number of disjoint components in random geometric graphs, which involves the Stein-Chen inequality for Poisson approximations~\cite{han2008connectivity}. This is different from the result of~\cite{galhotra2018connectivity}, which analyzes the connectivity of random annular graphs. Interestingly, our analysis also slightly improves the limit for unsupervised methods derived  in~\cite{galhotra2018connectivity}. % in some parameter regimes.
%This is another our additional theoretical contribution. \textcolor{blue}{It is worth mentioning that our algorithms will not query any labels above the limit we proved, which is good (don't know how to say "good" in a fancy way...).}

We test our algorithms extensively on both synthetic and real-world data. They improve the accuracy of the method in~\cite{galhotra2018geometric} from roughly $0.78$ to $0.92$ by querying no more than $4{\%}$ of the nodes in two real-world datasets. We remark that this accuracy is much higher than that of the spectral method, which can only achieve roughly $0.6$ accuracy on these same datasets. We also compare with the $S^2$ algorithm, which attains a slightly higher accuracy, but using at least $10$ times more queries.
%Interestingly, our results show that even using only $O(\log(n))$ queries we are able to make a large step forward beyond the limit stated in~\cite{galhotra2018connectivity}, where $n$ is the number of nodes in GBM.

\subsection{Related work}
% \Eli{I plan to discuss that transitivity can also be modeled by the following works 1) Abbe's works (Euclidean random graph), 2) hypergraph models (hSBM and more). I shall also mention the work of active learning on general graph/hypergraph ($S^2,HS^2$). Maybe also mention the work of yours if it's on Arxiv?}

\textbf{Active learning on arbitrary graphs--}
Active learning on graphs has attracted significant attention in the recent research literature. Most previous works do not assume any knowledge of the underlying statistical model, and hence their performance guarantees depend on the parameters of the graph into consideration~\cite{guillory2009label,gu2012towards,zhu2003combining,cesa2013active,dasarathy2015s2}. While these approaches are fairly general, they tend be too pessimistic in settings where prior knowledge about the statistical model is available.
In our work, we exploit the use of the minimax optimal graph-based active learning algorithm $S^2$~\cite{dasarathy2015s2} in combination with the prior knowledge of the underlying GBM.
%Hence our work bridges the field of active learning on arbitrary graph to the study of statistical graph model.

\textbf{Modeling transitivity--}
%Besides GBM, there are also other ways to model the transitivity. For example, the authors of~\cite{sankararaman2018community} study a different but similar model named Euclidean random graph.
Prior attempts to include transitivity in random graph models include the Euclidean random graph~\cite{sankararaman2018community}, where edges between nodes are randomly and independently drawn as a function of the distance between the corresponding nodes' feature random variables.
%In their model, edges are drawn between nodes randomly and independently as a function of the distance between the corresponding node feature random variables.
Differently from the GBM, clustering in this model requires, in addition to the graph, the values of the nodes' feature variables.
%The work of~\cite{sankararaman2018community} also considers the recovery scenario where in addition to the graph, values of the node features are provided. In contrast, the unsupervised method in GBM~\cite{galhotra2018connectivity,galhotra2018geometric} does not require this additional information.
Another transitivity driven model is the Gaussian mixture block model~\cite{abbe2018graph}, where node features are modeled via a Gaussian random vector with mean depending on the associated node label and identical variance.
%The mean of the Gaussian depends on the node labels and the variance are identical.
Two nodes are then connected by an edge if and only if their distance is smaller then some threshold. However, the authors of~\cite{abbe2018graph} only use this model to empirically validate their proposed unsupervised clustering method. No theoretical results have yet been proved for this model.

Finally, we note that while out of the scope of this paper, the use of hypergraphs provides another way to model transitivity, and that recent works have studied the generalization of the SBM in the hypergraph setting~\cite{chien2018community,chien2018minimax,ghoshdastidar2017consistency,ahn2016community,paul2018higher}.
%However, these works are obviously different from ours since we focus on the standard graph case.

\section{Notation and the geometric block model}

We use boldface upper case letters $\mathbf{A}$ to denote matrices and $[n]$ to denote the discrete set $\{1,2,...,n\}$. We use the standard asymptotic notation $f(n) = O(g(n))$ to denote that $\lim\limits_{n\rightarrow \infty}|\frac{f(n)}{g(n)}| \leq C$ for some constant $C\geq 0$. We also use $f(n) = o(g(n))$ to denote that $\lim\limits_{n\rightarrow \infty}|\frac{f(n)}{g(n)}| = 0$.
%We use $a \triangleq b$ to denote that we define $a$ to be equal to $b$.

We start by introducing the definition of the random geometric graph (RGG) model, which appeared as an alternative to the popular Erd{\"o}s-Renyi graph.
\begin{definition}[RGG, $2$ dimensional torus case]\label{def:RGG}
    A random graph under $RGG(n,r)$ is a graph with $n$ nodes, where each node $u\in [n]$ is associated with a latent feature vector $X_u\sim Unif[0,1]$. Letting the distance between $X_u$ and $X_v$ be defined as %$d(X_u,X_v)
    $d_{uv}= \min(|X_u-X_v|,1-|X_u-X_v|)$,
    % and we use $d_{uv}$ as an abbreviation.
    then, nodes $u,v$ are connected by an edge under $RGG(n,r)$ if and only if $d_{uv}\leq r$.
\end{definition}

Let $r \triangleq \frac{\theta\log(n)}{n}$ for some constant $\theta$. It is known that a random graph under $RGG(n,r)$ is connected with high probability if and only if $\theta > 1$~\cite{penrose2003random}. Next, we provide the definition of the GBM, which depends on the RGG in a similar manner as the SBM depends on the Erd{\"o}s-Renyi graph.
\begin{definition}[GBM, $2$ dimensional torus case~\cite{galhotra2018geometric,galhotra2018connectivity}]\label{def:GBM}
%Let an undirected graph $G = (V,E)$ associate with an adjacency matrix $\mathbf{A}$. If $G$ is generated from
A random graph under GBM$(n,\sigma,\theta_1,\theta_2)$ is a graph $G = (V,E)$ such that $V = [n]$ can be partitioned into two equal size components $V_1$ and $V_2$ determined by the label assignment $\sigma$. Specifically, $\sigma(i) = j$ if and only if $i\in V_j,\;\forall i\in[n],j=1,2$. Each node $u\in V$ is associated with a feature vector $X_u\sim Unif[0,1]$ independently from each other. Letting the distance between $X_u$ and $X_v$ be defined as %$d(X_u,X_v)
$d_{uv}= \min(|X_u-X_v|,1-|X_u-X_v|)$, %with $d_{uv}$ as an abbreviation. For any node pair $u,v$,
then, $(u,v)\in E$ if and only if
    $
        d_{uv}\leq \left(\theta_1\mathbf{1}\{\sigma(u) = \sigma(v)\}+\theta_2\mathbf{1}\{\sigma(u) \neq \sigma(v)\}\right)\frac{\log(n)}{n},
    $
    where $\theta_1 \geq \theta_2$ are constants independent of $n$.
\end{definition}

\begin{remark}
\label{gbmandrgg}
Note that each cluster in GBM$(n,\sigma,\theta_1,\theta_2)$ can be seen as a $RGG(n/2,\theta_1 \frac{\log(n)}{n})$.
\end{remark}

\begin{remark}
\label{scaling}
Note that we focus on the $\frac{\log(n)}{n}$ scaling regime as this is the critical regime for unsupervised exact recovery. As shown in~\cite{galhotra2018geometric,galhotra2018connectivity}, if $\theta_1-\theta_2<0.5$ or $\theta_1 < 1$, then no unsupervised method can correctly recover the community memberships with high probability.
%Although the more general definition of the GBM is also introduced in~\cite{galhotra2018connectivity}, their main focus regarding to the GBM is still for the $2$ dimensional torus case.
Thus, in the rest of the paper focus on the most relevant setting of Definition~\ref{def:GBM}.
\end{remark}

Note that in the GBM, a pair of nodes from the same community are connected with probability $\frac{2\theta_1\log(n)}{n}$, and nodes from different communities are connected with probability $\frac{2\theta_2\log(n)}{n}$. We refer to these two probabilities as the marginal distributions of the GBM. Similarly,
%along with the same definition we state in the introduction,
the marginal distributions of the SBM are $\frac{a\log(n)}{n},\frac{b\log(n)}{n}$.
%\Eli{Need some more polishing?}

\subsection{Limits of unsupervised learning in the GBM and in the SBM}
\label{limits}

In this section, we compare the limits of unsupervised clustering on SBM and GBM by
%, showing We observe that the limit is much higher in GBM when we
setting the marginal distributions of both models to be the same, and show how clustering in the GBM is fundamentally harder than in the SBM.

We first focus on the GBM. In order to achieve exact recovery, the algorithm of~\cite{galhotra2018connectivity} requires the parameters of the GBM to satisfy certain sophisticated constraints. Due to space limitations, we only list Table~\ref{tbl:abpair} for some examples of GBM parameter values that satisfy such constraints. The complete description of the corresponding theorem is stated in the Supplemental material.
\begin{table}[ht]
\centering
\begin{tabular}[t]{cccccc}
\toprule
 $\theta_2$ & 1 & 2 & 3 & 4 & 5\\
\midrule
  min $\theta_1$ & 8.96 & 12.63 & 15.9 & 18.98 & 21.93\\
\bottomrule
% \toprule
%  $\theta_2$ & 4 & 5 & 6 & 7\\
% \midrule
%   min $\theta_1$ & 18.98 & 21.93 & 24.78 & 27.57\\
% \bottomrule
\end{tabular}
% \vspace{0.1in}
\caption{The minimum $\theta_1$ for given $\theta_2$ such that the algorithm in~\cite{galhotra2018connectivity} would work.}\label{tbl:abpair}
\end{table}%
% Note that the probability of having an in-cluster edge is $\frac{2\theta_1\log(n)}{n}$ and having a cross-cluster edge is $\frac{2\theta_2\log(n)}{n}$ in GBM.

We now turn to the SBM. It is known that the state-of-the-art unsupervised method for the SBM requires $(\sqrt{a}-\sqrt{b})^2\geq 2$ to achieve exact recovery. We set $b = 2\theta_2$ and $a = 2\theta_1$ to make sure the marginal distributions are the same as in the GBM.
\begin{table}[ht]
\centering
\begin{tabular}[t]{cccccc}
\toprule
 $\frac{b}{2}$ & 1 & 2 & 3 & 4 & 5\\
\midrule
  min $\frac{a}{2}$ & 4 & 5.83 & 7.46 & 9 & 10.47\\
\bottomrule
% \toprule
%  $\theta_2$ & 4 & 5 & 6 & 7\\
% \midrule
%   min $\theta_1$ & 18.98 & 21.93 & 24.78 & 27.57\\
% \bottomrule
\end{tabular}
% \vspace{0.1in}
\caption{The minimum $a$ for given $b$ such that the best unsupervised method for SBM would work.}\label{tbl:abpair2}
\end{table}%

From Table~\ref{tbl:abpair} and \ref{tbl:abpair2}, we can observe that %the best known unsupervised method for GBM
exact recovery under the GBM
requires much denser connections within clusters than for the case of the SBM,
implying that clustering under the GBM is much harder than under the SBM.
This also means that many networks in practice, which are shown to follow the GBM more closely than the SBM~\cite{galhotra2018geometric}, will likely fall in the regimes where unsupervised methods cannot achieve exact recovery, further motivating the importance of active learning for community detection in real-world networks that exhibit transitivity.
%Thus,  under the assumption of GBM, it is more likely that a graph in practice does not satisfy the parameter constraints. This implies that clustering on GBM is much harder compare to the case of SBM which means active learning is even more important in GBM.

%Also note that GBM better models the real networks which is shown in~\cite{galhotra2018geometric}. Thus SBM might be a too optimistic model to study. These aspects motivate our study of active learning on GBM.

% Now let us first study the special case of our Spatial SBM, which paves the road for the analysis of our general model. Recall that by setting $a_1 = b_1 = 1$ and $a_2 = b_2 = 0$, we have the GBM introduced in~\cite{galhotra2018connectivity,galhotra2018geometric}. In both~\cite{galhotra2018geometric,galhotra2018connectivity}, they show that using triangle counting scheme can help us recover the communities. More specifically, the algorithm in~\cite{galhotra2018connectivity} can be understand as following. For each edge $(u,v)$, they count number of triangles $Z_{uv}$ that cover this edge. If $Z_{uv}\in [nE_L,nE_R]$ for some thresholds $E_L,E_R$ to be defined later, they will remove the edge $(u,v)$. They show that if we choose the thresholds properly then the algorithm returns two disjoint components which are exactly the underlying communities with high probability. The formal theoretical guarantee is as following

% In~\cite{galhotra2018connectivity}, they also give the following table for minimum value of $\theta_1$ given $\theta_2$.

\section{Active learning algorithms in the GBM}
% The main intuition that this algorithm would work is that choosing $t_1,t_2$ in~\eqref{req:t1234} ensures that we will remove all cross-cluster edges with high probability. However, this procedure will also remove some in-cluster edges. Galhotra et al. shows that the choice of $t_3,t_4$ in~\eqref{req:t1234} ensure the connectivity for each community with high probability.

% From Table~\ref{tbl:abpair} we can see that given $\theta_2 = 1$, the $\theta_1$ is required to be at least $8.96$ for the algorithm in~\cite{galhotra2018connectivity} to work. This requirement is quite strong since this implies that we need roughly the number of in-cluster edges to be $9$ times greater then the cross-cluster edges in expectation. A natural question thus arises: can we still recover the communities when $\theta_1$ is smaller than those in Table~\ref{tbl:abpair} with the help of a small number of label queries? We proposed an active learning algorithm which provides a smooth trade-off between the query complexity and the value of $\theta_1$. We prove that we will only need sub-linear numbers of queries to recover the communities with high probability even when $\theta_1$ is smaller than those stated in Table~\ref{tbl:abpair}. Moreover, we characterize the relation between $\theta_1,\theta_2$ and the exponent in the query complexity that our algorithm need.

In what follows, we present two active learning algorithms for the GBM, % that utilize the concept of motif-counting in different manners.
whose pseudocode is given described in Algrothm~\ref{alg:ALonGBM} and~\ref{alg:alter_ALonGBM}.
Both algorithms are composed of two phases: a first unsupervised phase that builds on the motif-counting technique of ~\cite{galhotra2018connectivity} to remove cross-cluster edges, and a second phase that queries a subset of node labels until recovering the underlying clusters.
%algorithms and they remove edges in the first phase based on motif-counting.

Phase 1 of Algorithm~\ref{alg:ALonGBM} removes as many cross-cluster edges as possible while preserving intra-cluster connectivity with high probability. During Phase 2, the $S^2$ algorithm is used to identify the remaining cross-cluster edges \footnote{For completeness, we include the $S^2$ algorithm in the Supplement}. In contrast, Algorithm~\ref{alg:alter_ALonGBM} adopts a more aggressive edge removing policy during Phase 1. That is, it removes all cross-cluster edges with high probability. Note that in this case, intra-cluster connectivity may no be preserved. Nevertheless, during Phase 2, querying the label of one node for each disjoint component is sufficient to recover the underlying clusters.

One of the key elements of Phase 1 in the proposed algorithms is the motif-counting technique used in~\cite{galhotra2018connectivity}. Here, a motif is simply defined as a configuration of triplets (triangles) in the graph. For any edge $(u,v)$, we count the number of triangles that cover edge $(u,v)$. It is shown in~\cite{galhotra2018connectivity} that this triangle count is statistically different when $\sigma(u) = \sigma(v)$ compared to $\sigma(u)\neq \sigma(v)$. More importantly, this count is also related to the distance of node features $d_{uv}$. We will discuss this more precisely in Section~\ref{sec:Analysis}.

\begin{algorithm2e}
\caption{Motif-counting with $S^2$}\label{alg:ALonGBM}
\SetAlgoLined
\DontPrintSemicolon
\SetKwInput{Input}{Input}
\SetKwInOut{Output}{Output}
\SetKwInput{Mainalgo}{Main Algorithm}
\SetKwInput{Phaseone}{Phase 1}
\SetKwInput{Phasetwo}{Phase 2}
  \Input{Graph $G=(V,E)$, threshold $E_T$.
  }
  \Output{Estimated labels $\hat{\sigma}$
  }
  Duplicate $G$ by $G_r$\;
  \Phaseone{\;
  \For{$(u,v)\in E$}{
  Calculate the number of triangles $T^{uv}$ that cover the edge $(u,v)$ on $G$, remove $(u,v)$ from $G_r$ if $T^{uv}\leq nE_T$.\;
  }}
  \Phasetwo{
  Apply $S^2$ to $G_r$ to get $\hat{\sigma}$. Terminate when we find $2$ disjoint components.
  }
\end{algorithm2e}
\begin{algorithm2e}
\caption{Aggressive edge removing approach}\label{alg:alter_ALonGBM}
\SetAlgoLined
\DontPrintSemicolon
\SetKwInput{Input}{Input}
\SetKwInOut{Output}{Output}
\SetKwInput{Mainalgo}{Main Algorithm}
\SetKwInput{Phaseone}{Phase 1}
\SetKwInput{Phasetwo}{Phase 2}
  \Input{Graph $G=(V,E)$, parameter $t_1$.
  }
  \Output{Estimated labels $\hat{\sigma}$
  }
  Duplicate $G$ by $G_r$\;
  \Phaseone{\;
  \For{$(u,v)\in E$}{
  Calculate the number of triangles $T^{uv}$ that cover the edge $(u,v)$ on $G$, remove $(u,v)$ from $G_r$ if $T^{uv}\leq (2\theta_2+t_1)\log(n)$.\;
  }}
  \Phasetwo{\;
  Query one node for each disjoint components in $G_r$ and assign labels according to queried nodes for each disjoint components.
  }
\end{algorithm2e}

In the following section, we show that under the assumption that $\theta_1 \geq 2 \theta_2$\footnote{Note that the condition $\theta_1 \geq 2 \theta_2$ is stronger than our model assumption $\theta_1 \geq \theta_2$} and $\theta_1 \geq 2$, both Algorithm~\ref{alg:ALonGBM} and Algorithm~\ref{alg:alter_ALonGBM} guarantee exact recovery with sub-linear query complexity.
However, note that if $\theta_1 <2$, the underlying clusters may already contain disconnected components and, consequently, Algorithm~\ref{alg:ALonGBM} may not be able to preserve intra-cluster connectivity, requiring additional queries to achieve exact recovery.
In this case, it is better to directly use Algorithm~\ref{alg:alter_ALonGBM} even if exact recovery with sub-linear query complexity can no longer be guaranteed.

Finally, in the numerical results of Section~\ref{sec:simulation}, we show that under the assumption of perfect knowledge of the underlying GBM, Algorithm~\ref{alg:alter_ALonGBM} has practically lower query complexity than Algorithm~\ref{alg:alter_ALonGBM}. However, when dealing with real datasets for which the parameters of the underlying GBM are not available, Algorithm~\ref{alg:ALonGBM} is shown to be more robust %easily adjusted to incorporate
to the uncertainty of the GBM parameters.

%{\RED
%\subsection{Algorithm Comparison}\label{sec:}
%Note that when $\theta_1 \geq 2$, both
%Algorithm~\ref{alg:ALonGBM} and Algorithm~\ref{alg:alter_ALonGBM} can be used. Under this condition and the assumption
%$\theta_1 \geq 2 \theta_2$\footnote{Note that the condition $\theta_1 \geq 2 \theta_2$ is stronger than our model assumption $\theta_1 \geq \theta_2$}, we prove that our algorithms guarantee exact recovery is a sub-linear query complexity in the next section.
%However if $\theta_1 <2$, the underlying clusters may already contain disconnected components and, consequently, Phase 1 of Algorithm~\ref{alg:ALonGBM} will not be able to preserve intra-cluster connectivity.
%In this case, it is better to directly use Algorithm~\ref{alg:alter_ALonGBM} even if exact recovery can not be guaranteed with sub-linear query complexity.
%Finally, under the assumption of perfect knowledge of the statistical model of the GBM, we observe that
%Algorithm~\ref{alg:alter_ALonGBM} has practically lower query complexity from the numerical results in Section~\ref{sec:simulation}. However,
%in the presence of real data where the knowledge on the statistical model is not available, Algorithm~\ref{alg:ALonGBM} can be more easily adjusted to incorporate the uncertainty in the statistical model.}
%This way in the next section we will mainly focus on a modified version of Algorithm~\ref{alg:ALonGBM}}

\section{Analysis of algorithms}\label{sec:Analysis}

In this section, we provide theoretical guarantees for our algorithms, and sketch the associated proofs. Detailed proofs are deferred to the supplementary material. We first state the result for the triangle count distribution.
\begin{lemma}[Lemma 11 and Lemma 12 in~\cite{galhotra2018connectivity}]\label{lma:Binomials}
    Assume $\theta_1 \geq 2\theta_2$. Let $\mathbf{A}$ be the adjacency matrix of GBM$(n,\sigma,\theta_1,\theta_2)$. For any pair of nodes $u,v$ with $A_{uv} = 1$, let $d_{uv} = x \triangleq \phi\frac{\log(n)}{n}$ and let the count of the triangles that cover edge $(u,v)$ be $T^{uv}(x) \triangleq \left|\left\{z\in V:A_{uz} = A_{vz} = 1\right\}\right|$. If $\sigma(u)\neq \sigma(v)$, then $$T^{uv}(x)\sim Bin(n-2,2\theta_2\frac{\log(n)}{n}).$$ If $\sigma(u)= \sigma(v)$, then
    \begin{align*}
        & T^{uv}(x)\sim Bin(\frac{n}{2}-2,(2\theta_1-\phi)\frac{\log(n)}{n})\\
        & +\mathbf{1}\{\phi\leq 2\theta_2\}Bin(\frac{n}{2},(2\theta_2-\phi)\frac{\log(n)}{n}).
    \end{align*}
\end{lemma}

Lemma~\ref{lma:Binomials} shows that indeed the triangle count is an informative metric to distinguish the cases of $\sigma(u) = \sigma(v)$ and $\sigma(u)\neq \sigma(v)$. It is also strongly related to the distance of node features $d_{uv}$. See Figure~\ref{fig:explain1} for the visualization.
% \begin{remark}
%     Note that given labels and $d_{uv}=x$, $T^{uv}_{z}(x) \triangleq \mathbf{1}\left\{A_{uz} = A_{vz} = 1\right\}$ is independent over $z$. This result is also proved in~\cite{galhotra2018connectivity,galhotra2018geometric}.
% \end{remark}

\subsection{Analysis of Algorithm~\ref{alg:ALonGBM}}
    We begin by stating the theoretical guarantee of Algorithm~\ref{alg:ALonGBM} under the assumption that $\theta_1 \geq 2 \theta_2$,
    %\footnote{Note that the condition $\theta_1 \geq 2 \theta_2$ is stronger than our model assumption $\theta_1 \geq \theta_2$}
    and  $\theta_1 \geq 2$.

\begin{theorem}\label{thm:ALonGBM}
Define
\scriptsize
    \begin{align}\label{req:eta}
            & t_1 = \inf\left\{t\geq 0:(2\theta_2+t)\log(\frac{2\theta_2+t}{2\theta_2})-t>1\right\},
    \end{align}
    \normalsize
    Under the assumption that  $\theta_1 > 2\theta_2 \geq 2$, set in  Algorithm~\ref{alg:ALonGBM}
    \scriptsize
    \begin{align}\label{req:eta0}
        & \eta = \inf\left\{t\geq 0:(\theta_1+\theta_2-2-t)\log(\frac{\theta_1+\theta_2-2-t}{\theta_1+\theta_2-2})+t>1\right\}\nonumber\\
        & E_T = (\theta_1+\theta_2-2-\eta)\frac{\log(n)}{n}.
    \end{align}
    \normalsize
    If $\theta_1-\theta_2-2-\eta>t_1$, then after \textbf{Phase 1} in Algorithm~\ref{alg:ALonGBM}, we already recover the communities up to a permutation with probability at least $1-o(1)$. If $t_1>\theta_1-\theta_2-2-\eta>0$, after \textbf{Phase 2},  with probability at least $1-o(1)$, Algorithm~\ref{alg:ALonGBM} will recover the communities up to a permutation with query complexity at most
    \begin{equation}
        O(n^{1-\epsilon}\log(n)^3+\log(n)),
        \label{complexity}
    \end{equation}
    where
    \scriptsize
    \begin{equation*}
        \epsilon = (\theta_1+\theta_2-2-\eta)\log(\frac{\theta_1+\theta_2-2-\eta}{2\theta_2})-(\theta_1-\theta_2-2-\eta).
    \end{equation*}
    \normalsize
    \end{theorem}

\begin{theorem}\label{thm:ALonGBM1}
Under the assumption that $2\theta_2 \leq 2$, $\theta_1 \geq 2$, setting
   \scriptsize
    \begin{align}%\label{req:eta0}
        & \eta = \inf\left\{t\geq 0:(2\theta_1-2-t)\log(\frac{2\theta_1-2-t}{2\theta_1-2})+t>2\right\} \nonumber\\
        & E_T = \frac{1}{2}(2\theta_1-2-\eta)\frac{\log(n)}{n},
    \end{align}
    \normalsize
the same theoretical guarantees for  Algorithm
~\ref{alg:ALonGBM}, stated in Theorem~\ref{thm:ALonGBM}, can be derived by redefining  $\epsilon$ in \eqref{complexity} as:
\scriptsize
$$\epsilon = (\frac{1}{2}(2\theta_1-2-\eta))\log(\frac{\frac{1}{2}(2\theta_1-2-\eta)}{2\theta_2})-(\frac{1}{2}(2\theta_1-2-\eta)-2\theta_2).$$
    \normalsize
\end{theorem}
\begin{remark}
%Note that under the assumption that $\theta_1 \geq 2 \theta_2$, with $\theta_1 \geq2$, the exponent $1-\epsilon$ in the query complexitydepends on the ``hardness'' of the parameters.
Note that for any fixed $\theta_2$ such that
$\theta_1 \geq 2 \theta_2$ with $\theta_1 \geq2$, $1-\epsilon$ decays as $\theta_1$ grows. Thus, Algorithm~\ref{alg:ALonGBM} provides a smooth trade-off between clustering hardness and query complexity. Interestingly, we show that when $\theta_1-\theta_2-2-\eta>t_1$, we can achieve exact recovery without any queries. We numerically show that this result gives an improvement over the previously known bound for unsupervised methods given in~\cite{galhotra2018connectivity} for a wide range of $\theta_2$ (See Figure~\ref{fig:improvement}).
\end{remark}

\begin{figure}[!t]
  \centering
    \includegraphics[width=0.85\linewidth]{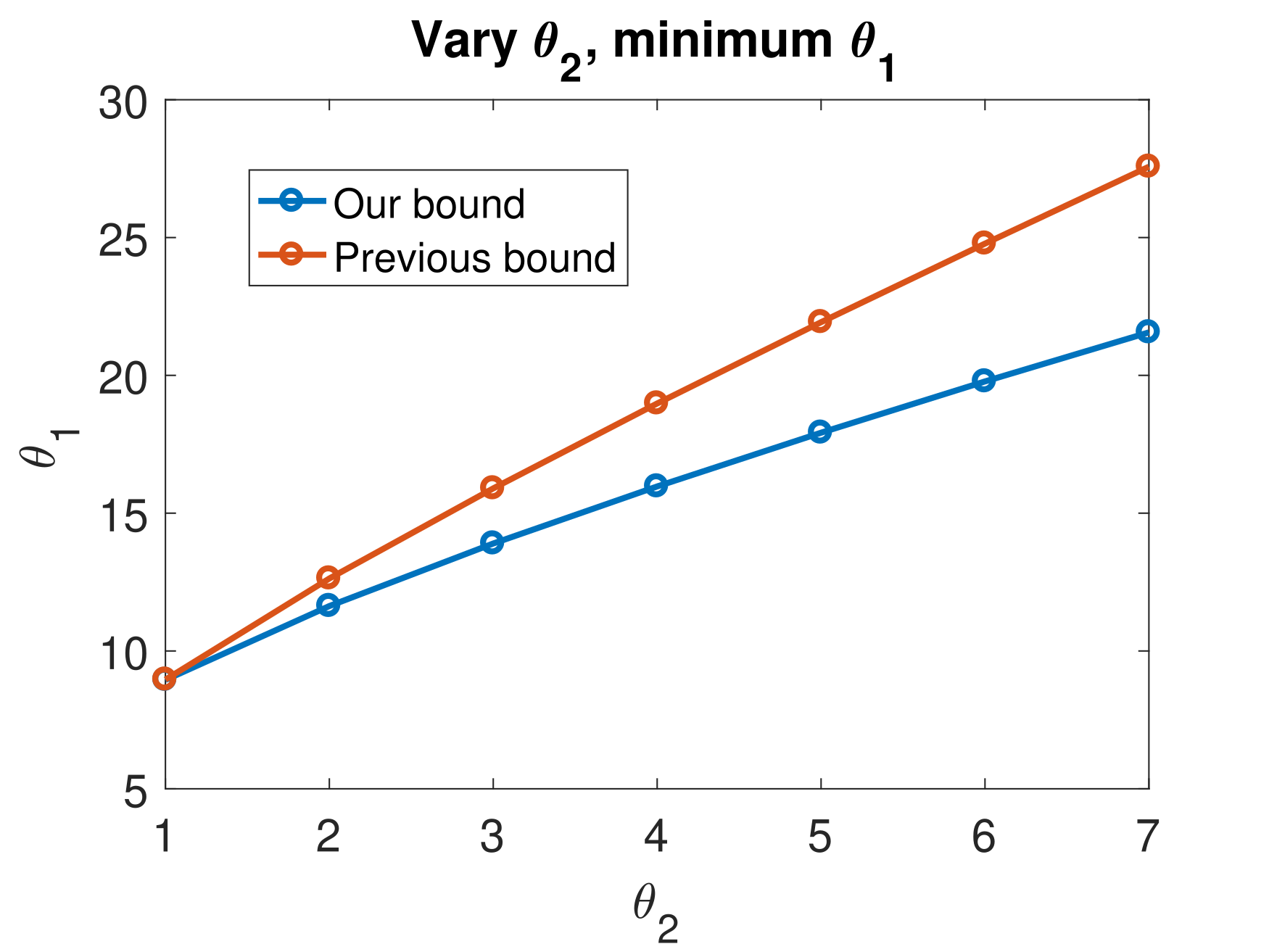}
  \caption{The minimum gap between $\theta_1$ and $\theta_2$ permitted by our Theorem~\ref{thm:ALonGBM} versus the state-of-the-art bound of~\cite{galhotra2018connectivity} in the unsupervised setting.}\label{fig:improvement}
\end{figure}

% Our results might be further improved when $\theta_2$ is some sufficiently large constant by using the connectivity result of random annular graph akin to~\cite{galhotra2018connectivity}. The analysis will be almost the same. \Eli{Not sure if we need to state this...}

%\begin{remark}
%    For the case $2\theta_2<2$, the analysis is even simpler. In this case we set
%    \scriptsize
%    $$\eta = \inf\left\{t\geq 0:(2\theta_1-2-t)\log(\frac{2\theta_1-2-t}{2\theta_1-2})+t>2\right\}$$ $$E_T = \frac{1}{2}(2\theta_1-2-\eta)\frac{\log(n)}{n}$$ $$\epsilon = (\frac{1}{2}(2\theta_1-2-\eta))\log(\frac{\frac{1}{2}(2\theta_1-2-\eta)}{2\theta_2})-(\frac{1}{2}(2\theta_1-2-\eta)-2\theta_2).$$
%    \normalsize
%    By similar analysis we can obtain the analogous result.
%\end{remark}
% \begin{remark}

% \end{remark}
% \begin{remark}
%     In the following, we will make use of the connectivity of each cluster which is determined are exactly the random geometric graph $RGG(\frac{n}{2},r\frac{\log(n)}{n})$. Note that the critical radius of $RGG(\frac{n}{2},r\frac{\log(n)}{n})$ is $$r\frac{\log(n)}{n} = \frac{\log(\frac{n}{2})}{\frac{n}{2}},$$
%     which means $r = 2-\frac{2\log(2)}{\log(n)} = 2-o(1)$. For simplicity we will set the critical value to be $2$ which is sufficient for the connectivity.
% \end{remark}

%\subsubsection{ \em Proof of Theorem~\ref{thm:ALonGBM}}

In the following, we focus on proving Theorem~\ref{thm:ALonGBM}, since Theorem~\ref{thm:ALonGBM1} can be proved analogously.
In order to prove Theorem~\ref{thm:ALonGBM}, we will use the theoretical guarantee of the $S^2$ algorithm and two technical lemmas.

\begin{theorem}[Simplified Theorem 3 in \cite{dasarathy2015s2}]\label{thm:S2}
    Let $C$ be the set of cross-cluster edges in graph $G$ with latent labels % associate with the labels
    $\mathbf{\sigma}$. Let $\partial C$ be the set of nodes associated with at least $1$ cross-cluster edge. Suppose that each cluster is connected and has diameter at most $D$. If the $S^2$ algorithm uses at least
    \scriptsize
    \begin{align*}
        \frac{\log(2/\delta)}{\log(2)} + \lceil\log_2(n)\rceil + (\min(|\partial C|,|C|)-1)(\lceil\log_2(2D+1)\rceil+1)
        \label{uffa}
    \end{align*}
    \normalsize
    queries, then with probability at least $1-\delta$ the $S^2$, the algorithm will recover the clusters exactly.
\end{theorem}
\begin{remark}
    We note that while the original analysis in~\cite{dasarathy2015s2} only uses the term $|\partial C|$ in the query complexity, the authors of~\cite{chien2019hs} slightly improve it by including the term $\min(|\partial C|,|C|)$, which better serves our analysis.
\end{remark}
\begin{lemma}\label{lma:eta_GBM}
    Assume $\theta_2 \geq 1$. Let
    \scriptsize
    \begin{equation*}
        \eta = \inf\left\{t\geq 0:(\theta_1+\theta_2-2-t)\log(\frac{\theta_1+\theta_2-2-t}{\theta_1+\theta_2-2})+t>1\right\}.
    \end{equation*}
    \normalsize
    Then, by choosing $E_T = (\theta_1+\theta_2-2-\eta)\frac{\log(n)}{n}$, {\bf Phase 1} of Algorithm~\ref{alg:ALonGBM} is guaranteed to generate a graph $G_r$ whose underlying communities are connected.
\end{lemma}
\begin{lemma}\label{lma:crossedgesbound_GBM}
  Assume $\theta_2 \geq 1$ and  set $E_T$ as in  Lemma~\ref{lma:eta_GBM}. Let $C$ be the set of cross-cluster edges in $G_r$. If $\theta_1-\theta_2-2-\eta>t_1,$, then with probability at least $1-o(1)$, we have $|C| = 0$. If $t_1>\theta_1-\theta_2-2-\eta>0,$,
    we have
    \begin{align*}
        |C| \leq \frac{\theta_2}{2}n^{1-\epsilon}(\log(n))^2
    \end{align*}
    with probability at least $1-o(1)$, where
    \begin{equation*}
        \epsilon = (\theta_1+\theta_2-2-\eta)\log(\frac{\theta_1+\theta_2-2-\eta}{2\theta_2})-(\theta_1-\theta_2-2-\eta)
    \end{equation*}
\end{lemma}

\begin{remark}
Note that while Lemma~\ref{lma:eta_GBM} characterizes the threshold in \textbf{Phase 1} of Algorithm~\ref{alg:ALonGBM} that guarantees removing the most cross-cluster edges while maintaining intra-cluster connectivity, Lemma~\ref{lma:crossedgesbound_GBM} provides a bound on the number of remaining cross-cluster edges. Such bound, together with the result stated in Theorem~\ref{thm:S2}, is one of the key ingredients in the evaluation of the query complexity bound of Algorithm~\ref{alg:ALonGBM}. %when we combine with the result of $S^2$ algorithm.
A graphical interpretation of the parameters
$t_1$ and $\eta$, as well as of the key steps of the proof of Lemma~\ref{lma:crossedgesbound_GBM} is provided in Figures~\ref{fig:explain1} and \ref{fig:explain2}.
%See also the Figure~\ref{fig:explain1} for the idea of . The Figure~\ref{fig:explain2} illustrates the key idea of the proof of Lemma~\ref{lma:crossedgesbound_GBM}.
\end{remark}

%\subsubsection{Proof of Theorem~\ref{thm:ALonGBM}}
We are now ready to provide the proof of Theorem~\ref{thm:ALonGBM}.
\begin{proof}
    The first half of Theorem~\ref{thm:ALonGBM} directly follows  from Lemma~\ref{lma:eta_GBM} and ~\ref{lma:crossedgesbound_GBM}. Hence, in the following, we focus on the case $t_1>\theta_1-\theta_2-2-\eta>0$. From Lemma~\ref{lma:eta_GBM}, we know that with probability at least $1-o(1)$ of the underlying clusters of graph $G_r$ (the graph returned by {\bf Phase 1}) are still connected among each other. Then, by Theorem~\ref{thm:S2}, we know that with probability at least $1-\delta$, using at most
    \scriptsize
    $$\frac{\log(2/\delta)}{\log(2)} + \lceil\log_2(n)\rceil + (\min(|\partial C|,|C|)-1)(\lceil\log_2(n+1)\rceil+1)$$
    \normalsize
    queries, we can recover the communities. Finally, by Lemma~\ref{lma:crossedgesbound_GBM}, we know that with probability at least $1-o(1)$, $\min(|\partial C|,|C|)\leq |C| \leq \frac{\theta_2}{2}n^{1-\epsilon}\log(n)^2$. Hence, by union bound over all error events, we know that with probability at least $1-o(1)$, we can recover the communities with at most
    $$O(n^{1-\epsilon}\log(n)^3+\log(n))$$ queries,
    by simply choosing $\delta = \frac{1}{n}$.
\end{proof}

\begin{figure*}[t]
  \centering
    \subfloat[Case for $\theta_1-\theta_2-2-\eta>t_1$\label{fig:explain1}]{\includegraphics[width=0.32\linewidth]{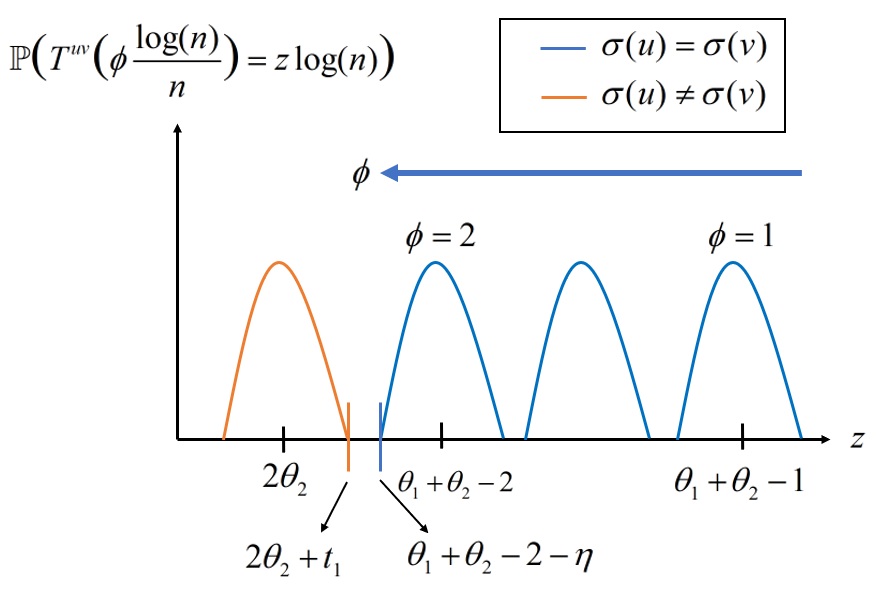}}
    \subfloat[Case for $\theta_1-\theta_2-2-\eta<t_1$. Key idea of Lemma~\ref{lma:crossedgesbound_GBM}\label{fig:explain2}]{\includegraphics[width=0.32\linewidth]{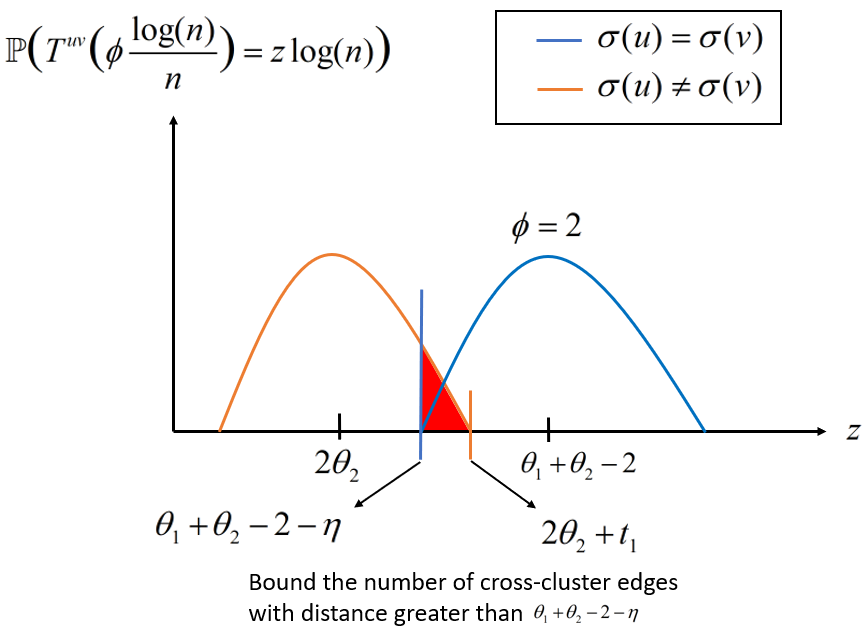}}
    \subfloat[Case for $\theta_1-\theta_2-2-\eta<t_1$. Key idea of Lemma~\ref{lma:lemmaR}\label{fig:explain3}]{\includegraphics[width=0.32\linewidth]{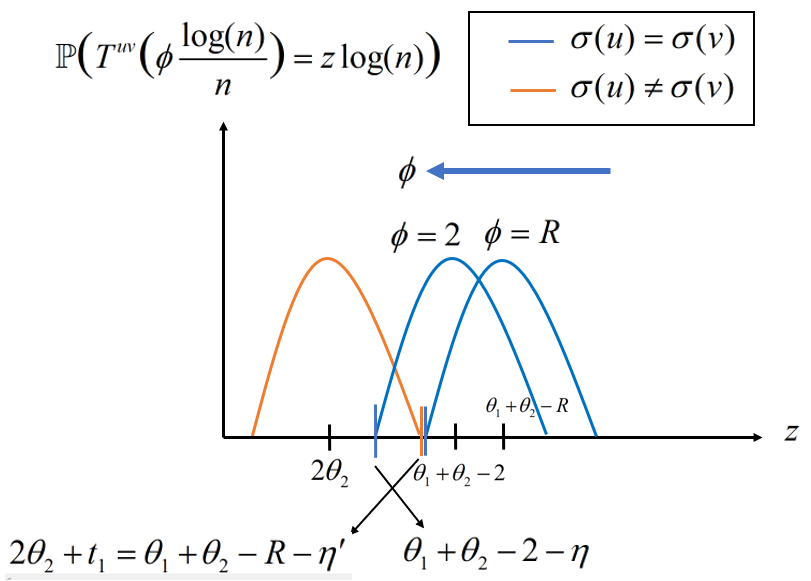}}
    % \subfigure[\label{fig:Q_GBM_b5}]{\includegraphics[width=0.48\linewidth]{Q_GBM_b5_new.eps}}
  \caption{
  %The illustration for the proof of theorems.
  Figure (a) illustrates Lemma~\ref{lma:Binomials} and the intuition behind parameters $t_1$ and $\eta$. Figure (b) illustrates the main idea behind the proof of Lemma~\ref{lma:crossedgesbound_GBM}. We use the error probability (red area) to compute the expected number of cross-cluster edges that are kept after \textbf{Phase 1} of Algorithm~\ref{alg:ALonGBM}. Then, we use Markov inequality to give the high probability bound for the number of cross-cluster edges in $G_r$. Figure (c) illustrate the idea of parameter $R$, which is chosen to be the largest number such that $T^{uv}(R\frac{\log(n)}{n})\geq 2\theta_2+t_1$ for intra-cluster edges (blue) with high probability. }\label{fig:explain}
\end{figure*}

\subsection{Analysis of Algorithm~\ref{alg:alter_ALonGBM}}
The next theorem provides the theoretical guarantee of Algorithm~\ref{alg:alter_ALonGBM} under the assumption that
$\theta_1 \geq 2 \theta_2$, and $\theta_2>1$.

\begin{theorem}\label{thm:QboundAlgo2}
    Assume $\theta_1 \geq 2 \theta_2$,  $\theta_2 \geq 1$, and $2\theta_2+t_1 > \theta_1 + \theta_2 - 2 -\eta$. With probability at least $1-o(1)$, Algorithm~\ref{alg:alter_ALonGBM} exactly recovers the underlying clusters with query complexity at most $$\frac{3}{2}n^{1-R/2}+2,$$ where
    \scriptsize
    \begin{align*}
        R = &\sup_{\min(\theta_1-\theta_2-t_1,2)>r>0} \Bigg\{ (2\theta_2+t_1)\log(\frac{2\theta_2+t_1}{\theta_1+\theta_2-r})\\
        &+(\theta_1+\theta_2-r-(2\theta_2+t_1)) > 1\Bigg\}.
    \end{align*}
    \normalsize
\end{theorem}

Note that if $2\theta_2+t_1 < \theta_1 + \theta_2 - 2 -\eta$, since Algorithm~\ref{alg:alter_ALonGBM} sets the threshold for the triangle count to $2\theta_2+t_1$, it is immediate to note (see Figure~\ref{fig:explain1}) that all cross-cluster edges will be removed while preserving all intra-cluster edges whose distance $\phi\frac{\log(n)}{n}$ is less than $2\frac{\log(n)}{n}$. In order to proof Theorem~\ref{thm:QboundAlgo2}, we need the following lemmas.
\begin{lemma}\label{lma:lemmaR}
    Assume $\theta_1 \geq 2\theta_2$,  $\theta_2 \geq 1$ and $2\theta_2+t_1 > \theta_1 + \theta_2 - 2 -\eta$. All intra-cluster edges with distance less than $R$ will not be removed in $G_r$, where
    \scriptsize
    \begin{align*}
       R = &\sup_{\min(\theta_1-\theta_2-t_1,2)>r>0} \Bigg\{ (2\theta_2+t_1)\log(\frac{2\theta_2+t_1}{\theta_1+\theta_2-r})\\
       &+(\theta_1+\theta_2-r-(2\theta_2+t_1)) > 1\Bigg\}
            %  & y = (\frac{r}{2}-1)\log(n)-\log(2)
    \end{align*}
    \normalsize
\end{lemma}

Figure~\ref{fig:explain3} provides a graphical illustration of $R$. The key idea is to find the largest $R$ such that all intra-cluster edges with distance smaller than $R$ will not be removed with high probability during \textbf{Phase 1} of Algorithm~\ref{alg:alter_ALonGBM}.

Next, we characterize the number of disjoint components created in each cluster by
\textbf{Phase 1} of Algorithm~\ref{alg:alter_ALonGBM}.
To this end (see Remark~\ref{gbmandrgg}), we resort to the following lemma.
% Next we need the following lemma which shows that the number of disjoint components is approximately a Poisson random variable.
\begin{lemma}[Modification of Theorem 8.1 in~\cite{han2008connectivity}]\label{lma:Poissonapprox}
 Given a random geometric graph $RGG(n,\tau)$ with $2\tau<1$,  let $\mathcal{C}_{n,\tau}$ be the probability mass  function of the number of  disjoint components $-1$ of $RGG(n,\tau)$, and let $\Pi_\lambda$ denote a Poisson distribution with parameter $\lambda$. Let $d_{TV}(\mu,\nu)\triangleq \frac{1}{2}\sum_{x=0}^{\infty}|\mu(x)-\nu(x)|$ be the total variation of the two  probability mass functions $\mu$ and $\nu$ on $\mathbb{N}$. We then have
    \begin{align*}
        d_{TV}(\mathcal{C}_{n,\tau},\Pi_{\lambda_n(\tau)}) \leq B_n(\tau),
    \end{align*}
    where
    \scriptsize
    \begin{align*}
        & \lambda_n(\tau) = n(1-\tau)^{n},\;B_n(\tau) = n(1-\tau)^{n} - (n-1)(1-\frac{\tau}{1-\tau})^{n}
    \end{align*}
    \normalsize
\end{lemma}

The key idea of the proof of Theorem 8.1 in~\cite{han2008connectivity} is observing that the number of disjoint components can be related to the indicator functions of the spacing of uniform random variables.
Note that these indicator functions are nothing but properly correlated Bernoulli random variables which can be approximated by a Poisson random variable where the total variation can be bounded via Stein-Chen's method~\cite{chen1975poisson}. See more details in~\cite{han2008connectivity}, and the modification for our setting in the supplementary material.

%Finally, we use the following lemma to bound the tail probability of Poisson random variable.
\begin{lemma}[Poisson tail bound~\cite{ref:poissontailbound}]\label{lma:Poisson_tail}
   Let $X \sim \Pi_{\lambda}$ be  a Poisson random variable with parameter $\lambda$. For all $y>0$, $$\mathbb{P}\left(X \geq \lambda + y \right) \leq \exp \left (-\frac{y^2}{2(\lambda+y)} \right).$$
\end{lemma}

We are now ready to state the sketch of the proof of Theorem~\ref{thm:QboundAlgo2}.
First, we use Lemma~\ref{lma:lemmaR} to find the largest distance $R$ such that,  with high probability, all intra-cluster edges with distance smaller than $R$  will not be removed during \textbf{Phase 1} of Algorithm~\ref{alg:alter_ALonGBM}. Next, we note that the number of disjoint components in $G_r$ can be upper bounded by twice the number of disjoint components in an RGG$(\frac{n}{2},R\frac{\log(n)}{n})$. Using Lemma~\ref{lma:Poissonapprox}, we approximate the number of disjoint components in RGG$(\frac{n}{2},R\frac{\log(n)}{n})$ as a Poisson random variable with parameter given in Lemma~\ref{lma:Poissonapprox}.
%using Lemma~\ref{lma:Poisson_tail} gives a tight tail bound for Poisson distribution. Thus
Combining Lemma~\ref{lma:Poissonapprox} and \ref{lma:Poisson_tail}, we are able to establish an upper bound on the number of disjoint components in $G_r$, which directly leads to the query complexity bound. The rigorous proof of Theorem~\ref{thm:QboundAlgo2} is deferred to the supplementary material.

\section{Experimental Results}\label{sec:simulation}
\subsection{Synthetic Datasets}
\begin{figure*}[!t]
  \centering
    % \subfigure[\label{fig:Q_GBM_b3}]{\includegraphics[width=0.32\linewidth]{Q_GBM_b3_new.eps}}
    \subfloat[\label{fig:Q_GBM_b4}]{\includegraphics[width=0.42\linewidth]{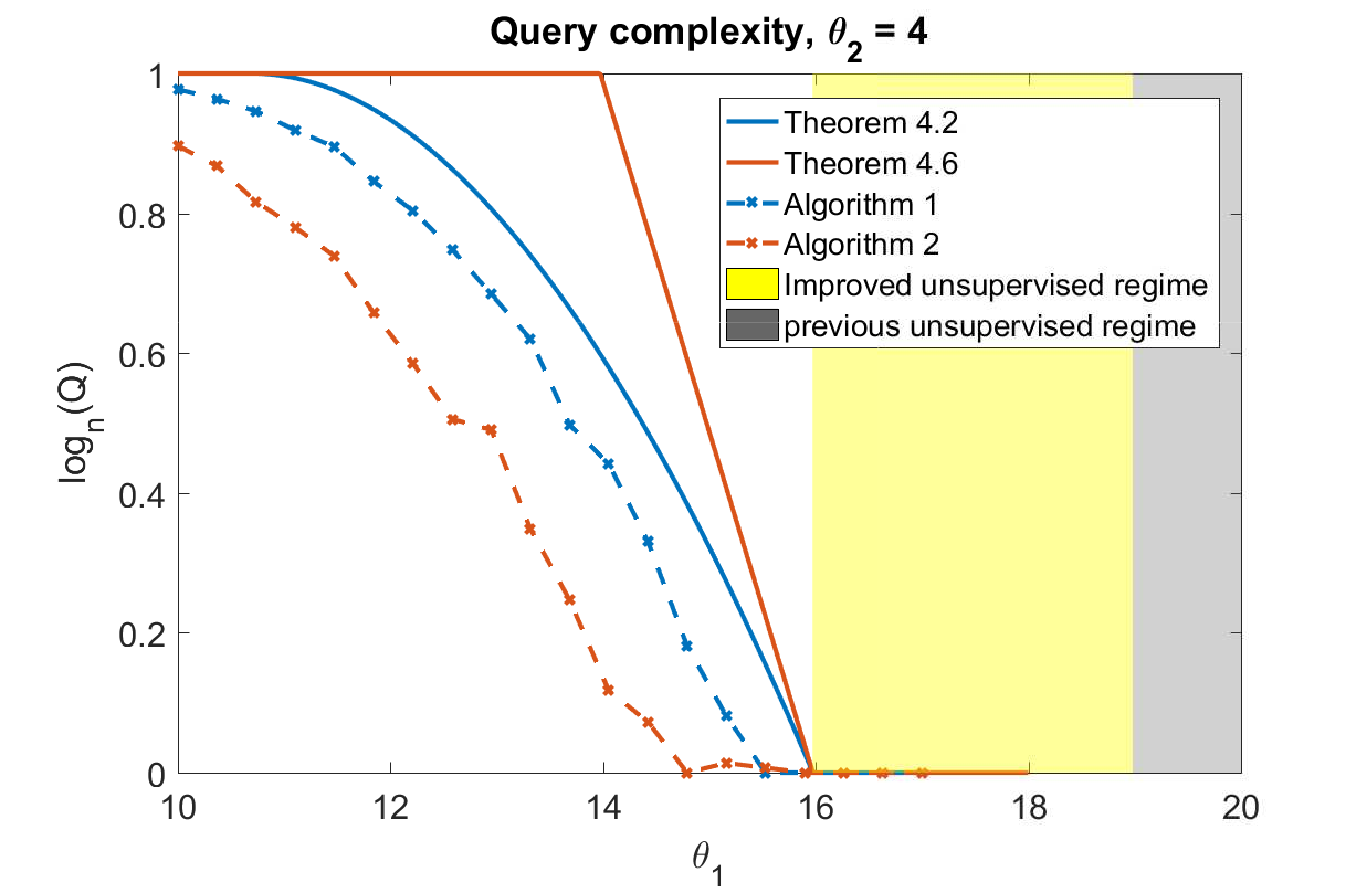}}
    \subfloat[\label{fig:Q_GBM_b5}]{\includegraphics[width=0.42\linewidth]{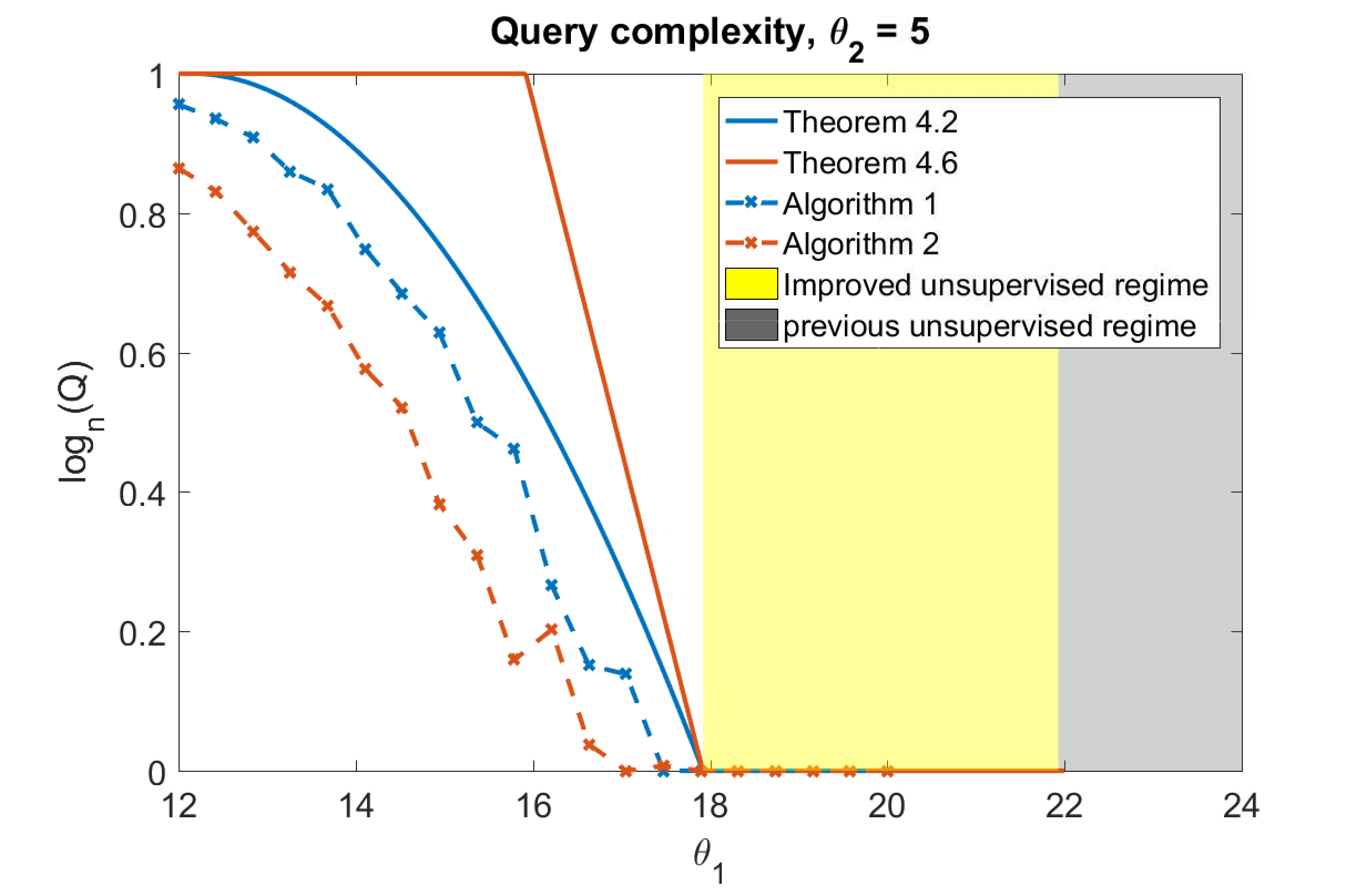}}
  \caption{Query Complexity of our active learning algorithms in the GBM, where we use $Q$ to denote the query complexity. Results are averaged over $20$ independent trials.  The light yellow shaded area indicates the improvement of our approach compared with ~\cite{galhotra2018connectivity}(grey shaded area) in the unsupervised setting. For Theorem~\ref{thm:ALonGBM} and~\ref{thm:QboundAlgo2}, we only plot the main term in the theoretical bounds, $n^{1-\epsilon}$ and $n^{1-R/2}$.}\label{fig:Q_GBM}
\end{figure*}

We generate random graphs using a GBM$(n,\sigma,\theta_1,\theta_2)$ where $n = 1000$ and $\sigma$ is chosen arbitrarily among the equal-size community assignment. We plot the query complexity as a function of $\theta_1$ for some fixed $\theta_2$ in Figure~\ref{fig:Q_GBM}. The figures for the other choices of $\theta_2$ are deferred to the supplementary material.
Figure~\ref{fig:Q_GBM} plots the logarithm of the query complexity (i.e $\log_n(Q)$) as a function of $\theta_1$ for a given $\theta_2$.\footnote{Note that $\log_n(Q)<1$ implies that a sub-linear number of queries can achieve exact recovery.}
%which means we plot the regime
Figure~\ref{fig:Q_GBM} shows that
our results improve the previously known bound for unsupervised methods given in~\cite{galhotra2018connectivity} for a wide range of $\theta_2$, as already stated in Section~\ref{sec:Analysis}, and our theorems capture the behavior of the query complexity for both Algorithm~\ref{alg:ALonGBM} and ~\ref{alg:alter_ALonGBM}.
%{ \RED Also note that our results offer a smooth trade-off between query complexity and the hardness of clustering.}

As expected, from Figure~\ref{fig:Q_GBM},
 we observe that for a fixed $\theta_2$, as $\theta_1$ decreases, we need more queries in order to achieve exact recovery. This number of queries is a very little fraction of the total number nodes, especially as $n$ grows large. For example, in our experiment we set $n = 1000$. Then, using $n^{0.5}$ queries means we only use around $32$ queries, which is just $3 \%$ of the total number of nodes.

Interestingly, Theorem~\ref{thm:ALonGBM} offers a lower query complexity bound compared with Theorem~\ref{thm:QboundAlgo2}. However, from Figure~\ref{fig:Q_GBM}, we can see that Algorithm~\ref{alg:alter_ALonGBM} has a lower query complexity in practice, which implies that our Theorem~\ref{thm:QboundAlgo2} can be improved. In fact, in our analysis of the number of disjoint components, we only take into account the edges that with high probability will not be removed.  However, there are still some edges which will be removed only with constant probability, and each of them could potentially reduce the number of disjoint components by $1$. Nevertheless, this analysis is much more complicated since these edges are not independent, and hence is left for future work.
%It would be interesting if Theorem~\ref{thm:QboundAlgo2} can be further improved in this direction.

%has a lower query complexity

%From the We also find that Algorithm~\ref{alg:alter_ALonGBM} depends more heavily on the model assumption hence we will only focus on the modified Algorithm~\ref{alg:ALonGBM} in this section.

\subsection{Real Datasets}
We use the following real datasets:
\begin{itemize}
    \item \textbf{Political Blogs (PB):}~\cite{adamic2005political} It contains a list of political blogs from the 2004 US Election classified as liberal or conservative, and links between blogs. The clusters are of roughly the same size $(586,636)$ with a total of $1222$ nodes and $16714$ edges.
    \item \textbf{LiveJournal (LJ):}~\cite{yang2015defining} The LiveJournal dataset is a free online blogging social network of around $4$ million users. We extract the top two clusters of sizes (1430, 936) which consist of around $11.5$K edges.
\end{itemize}

\textbf{Experimental Setting:} For real-world networks $G = (V,E)$, it is hard to obtain an exact threshold as the actual values of $\theta_1$ and $\theta_2$ are unknown. Hence, following the idea proposed in~\cite{galhotra2018geometric}, we use a similar but much more intuitive approach compared with~\cite{galhotra2018geometric}, which consists of 3 phases.
%In the first phase,
%we use a motif-counting to identify a large connected subgraph $\mathcal{V}_0$. In the second phase we classify the nodes of $\mathcal{V}_0$ using the $S^2$ algorithm. Finally in the third phase we decide on the class of all nodes that that do not belong the subgraph $\mathcal{V}_0$. Specifically,
In the first phase, we set a threshold $T_1$.  We remove all edges $(u,v)$  covered by less than $T_1$ triangles, and we identify $\mathcal{V}_0$  as the largest connected component of the obtained graph. In the second phase, we directly apply the $S^2$ algorithm on $\mathcal{V}_0$ and terminate it when we find $2$ non-singleton disjoint components in $\mathcal{V}_0$. Finally, in the third phase, we take majority voting to decide the cluster of $w\in  V \setminus \mathcal{V}_0$ based on all $u \in \mathcal{V}_0$ such that the edge $(u,w)$ exists.
%for each node that is not part of $\mathcal{V}_0$, say $w \in  V \setminus \mathcal{V}_0$, for all $u \in \mathcal{V}_0$ such that the edge $(u,w)$ exists, we take majority voting to decide the cluster of $w$. This strategy is inspired by our Algorithm~\ref{alg:ALonGBM}}.
%First we also use a somewhat large threshold $T_1$ to sample a subgraph $S$. We remove all edges $(u,v)$ such that $(u,v)$ is covered by less than $T_1$ triangles. We extract the largest connected component to get $S$. Then we directly apply $S^2$ algorithm on $S$. We terminate $S^2$ algorithm when we find $2$ non-singleton disjoint components in $S$. Finally, for nodes that are not part of $S$, say $w\in V\setminus S$, we take majority voting among its neighborhood in $S$ to decide the the cluster of $w$. This strategy is inspired by our Algorithm~\ref{alg:ALonGBM}.
%Note that compare to the unsupervised method used in~\cite{galhotra2018geometric}, they additionally need $2$ more hyperparameters $T_2$ and $T_3$ to achieve their stated result. In contrast, our active learning method here only need 1 hyperparameter $T_1$.
Note that, in contrast with  the unsupervised method used in~\cite{galhotra2018geometric}, where two more hyperparameters $T_2$ and $T_3$ are required,
our active learning method only needs one hyperparameter $T_1$.
We use GMPS18 to denote the unsupervised method in~\cite{galhotra2018geometric} and Spectral to denote the standard spectral method.
%The performance of these two methods are directly obtained from~\cite{galhotra2018geometric}.
All results are averaged over $100$ independent trials.
\footnotesize
\begin{table}[h]
\centering
  \begin{tabular}{lSSSS}
    \toprule
    \multirow{2}{*}{Method} &
      \multicolumn{2}{c}{Accuracy} &
      \multicolumn{2}{c}{Query complexity ($\%$)} \\
      & {PB} & {LJ} & {PB} & {LJ} \\
      \midrule
         Ours &  0.931 & 0.912 & 3.7$\%$ & 0.88$\%$\\
        \midrule
         Spectral & 0.53 & 0.64 & 0 & 0 \\
        \midrule
         GMPS18 & 0.788 & 0.777 & 0 & 0  \\
         \midrule
         $S^2$ & 0.97 & 0.999 & 47.2$\%$ & 9.2 $\%$\\
        \bottomrule
  \end{tabular}
  \caption{Performance on real-world dataset. }\label{tbl:realdata}
\end{table}
\normalsize

For our method, we choose $T_1 = 30$ for the Political Blogs dataset and $T_1=5$ for the LiveJournal dataset. From Table~\ref{tbl:realdata}, we can see that our active learning method only queries $3.7\%$ of nodes and significantly improves the accuracy from $0.788$ to $0.931$ in the Political Blogs dataset. Also, note that if we directly apply $S^2$ without using triangle counting, it will query 47.2$\%$ of nodes before termination. Apparently, this is too expensive in terms of query complexity. A similar result can also be observed on the LiveJournal dataset. Hence, combining triangle counting is necessary for obtaining a practical solution in the active learning framework when we have a limited query budget.

\section*{Acknowledgments}
The authors thank Sainyam Galhotra for providing the experiment details on real datasets. 

\small
\bibliography{example_paper}

\begin{thebibliography}{}

\bibitem[\protect\citeauthoryear{Abbe \bgroup et al\mbox.\egroup
  }{2018}]{abbe2018graph}
Abbe, E.; Boix, E.; Ralli, P.; and Sandon, C.
\newblock 2018.
\newblock Graph powering and spectral robustness.
\newblock {\em arXiv preprint}.

\bibitem[\protect\citeauthoryear{Abbe}{2017}]{abbe2017community}
Abbe, E.
\newblock 2017.
\newblock Community detection and stochastic block models: recent developments.
\newblock {\em The Journal of Machine Learning Research} 18(1):6446--6531.

\bibitem[\protect\citeauthoryear{Adamic and Glance}{2005}]{adamic2005political}
Adamic, L.~A., and Glance, N.
\newblock 2005.
\newblock The political blogosphere and the 2004 us election: divided they
  blog.
\newblock In {\em Proceedings of the 3rd international workshop on Link
  discovery},  36--43.
\newblock ACM.

\bibitem[\protect\citeauthoryear{Ahn, Lee, and Suh}{2016}]{ahn2016community}
Ahn, K.; Lee, K.; and Suh, C.
\newblock 2016.
\newblock Community recovery in hypergraphs.
\newblock In {\em 2016 54th Annual Allerton Conference on Communication,
  Control, and Computing (Allerton)},  657--663.
\newblock IEEE.

\bibitem[\protect\citeauthoryear{Cesa-Bianchi \bgroup et al\mbox.\egroup
  }{2013}]{cesa2013active}
Cesa-Bianchi, N.; Gentile, C.; Vitale, F.; and Zappella, G.
\newblock 2013.
\newblock Active learning on trees and graphs.
\newblock {\em arXiv preprint}.

\bibitem[\protect\citeauthoryear{Chen and Yuan}{2006}]{chen2006detecting}
Chen, J., and Yuan, B.
\newblock 2006.
\newblock Detecting functional modules in the yeast protein--protein
  interaction network.
\newblock {\em Bioinformatics} 22(18):2283--2290.

\bibitem[\protect\citeauthoryear{Chen}{1975}]{chen1975poisson}
Chen, L.~H.
\newblock 1975.
\newblock Poisson approximation for dependent trials.
\newblock {\em The Annals of Probability}  534--545.

\bibitem[\protect\citeauthoryear{Chien, Lin, and
  Wang}{2018}]{chien2018community}
Chien, I.; Lin, C.-Y.; and Wang, I.-H.
\newblock 2018.
\newblock Community detection in hypergraphs: Optimal statistical limit and
  efficient algorithms.
\newblock In {\em International Conference on Artificial Intelligence and
  Statistics},  871--879.

\bibitem[\protect\citeauthoryear{Chien, Lin, and Wang}{2019}]{chien2018minimax}
Chien, I.; Lin, C.-Y.; and Wang, I.-H.
\newblock 2019.
\newblock On the minimax misclassification ratio of hypergraph community
  detection.
\newblock {\em IEEE Transactions on Information Theory}.

\bibitem[\protect\citeauthoryear{Chien, Zhou, and Li}{2019}]{chien2019hs}
Chien, I.~E.; Zhou, H.; and Li, P.
\newblock 2019.
\newblock $ {HS^2}$: Active learning over hypergraphs with pointwise and
  pairwise queries.
\newblock In {\em The 22nd International Conference on Artificial Intelligence
  and Statistics},  2466--2475.

\bibitem[\protect\citeauthoryear{{Cl{\'e}ment
  Canonne}}{2019}]{ref:poissontailbound}
{Cl{\'e}ment Canonne}.
\newblock 2019.
\newblock A short note on poisson tail bounds.

\bibitem[\protect\citeauthoryear{Cohn, Atlas, and
  Ladner}{1994}]{cohn1994improving}
Cohn, D.; Atlas, L.; and Ladner, R.
\newblock 1994.
\newblock Improving generalization with active learning.
\newblock {\em Machine learning} 15(2):201--221.

\bibitem[\protect\citeauthoryear{Dasarathy, Nowak, and
  Zhu}{2015}]{dasarathy2015s2}
Dasarathy, G.; Nowak, R.; and Zhu, X.
\newblock 2015.
\newblock S2: An efficient graph based active learning algorithm with
  application to nonparametric classification.
\newblock In {\em Conference on Learning Theory},  503--522.

\bibitem[\protect\citeauthoryear{Devroye \bgroup et al\mbox.\egroup
  }{2011}]{devroye2011high}
Devroye, L.; Gy{\"o}rgy, A.; Lugosi, G.; Udina, F.; et~al.
\newblock 2011.
\newblock High-dimensional random geometric graphs and their clique number.
\newblock {\em Electronic Journal of Probability} 16:2481--2508.

\bibitem[\protect\citeauthoryear{Fortunato}{2010}]{fortunato2010community}
Fortunato, S.
\newblock 2010.
\newblock Community detection in graphs.
\newblock {\em Physics reports} 486(3-5):75--174.

\bibitem[\protect\citeauthoryear{Gadde \bgroup et al\mbox.\egroup
  }{2016}]{gadde2016active}
Gadde, A.; Gad, E.~E.; Avestimehr, S.; and Ortega, A.
\newblock 2016.
\newblock Active learning for community detection in stochastic block models.
\newblock In {\em 2016 IEEE International Symposium on Information Theory
  (ISIT)},  1889--1893.
\newblock IEEE.

\bibitem[\protect\citeauthoryear{Galhotra \bgroup et al\mbox.\egroup
  }{2018}]{galhotra2018geometric}
Galhotra, S.; Mazumdar, A.; Pal, S.; and Saha, B.
\newblock 2018.
\newblock The geometric block model.
\newblock In {\em Thirty-Second AAAI Conference on Artificial Intelligence}.

\bibitem[\protect\citeauthoryear{Galhotra \bgroup et al\mbox.\egroup
  }{2019}]{galhotra2018connectivity}
Galhotra, S.; Mazumdar, A.; Pal, S.; and Saha, B.
\newblock 2019.
\newblock Connectivity in random annulus graphs and the geometric block model.
\newblock {\em The International Conference on Randomization and Computation}.

\bibitem[\protect\citeauthoryear{Ghoshdastidar, Dukkipati, and
  others}{2017}]{ghoshdastidar2017consistency}
Ghoshdastidar, D.; Dukkipati, A.; et~al.
\newblock 2017.
\newblock Consistency of spectral hypergraph partitioning under planted
  partition model.
\newblock {\em The Annals of Statistics} 45(1):289--315.

\bibitem[\protect\citeauthoryear{Goel \bgroup et al\mbox.\egroup
  }{2005}]{goel2005monotone}
Goel, A.; Rai, S.; Krishnamachari, B.; et~al.
\newblock 2005.
\newblock Monotone properties of random geometric graphs have sharp thresholds.
\newblock {\em The Annals of Applied Probability} 15(4):2535--2552.

\bibitem[\protect\citeauthoryear{Gu and Han}{2012}]{gu2012towards}
Gu, Q., and Han, J.
\newblock 2012.
\newblock Towards active learning on graphs: An error bound minimization
  approach.
\newblock In {\em 2012 IEEE 12th International Conference on Data Mining},
  882--887.
\newblock IEEE.

\bibitem[\protect\citeauthoryear{Guillory and Bilmes}{2009}]{guillory2009label}
Guillory, A., and Bilmes, J.~A.
\newblock 2009.
\newblock Label selection on graphs.
\newblock In {\em Advances in Neural Information Processing Systems},
  691--699.

\bibitem[\protect\citeauthoryear{Gupta and Kumar}{1999}]{gupta1999critical}
Gupta, P., and Kumar, P.~R.
\newblock 1999.
\newblock Critical power for asymptotic connectivity in wireless networks.
\newblock In {\em Stochastic analysis, control, optimization and applications}.
  Springer.
\newblock  547--566.

\bibitem[\protect\citeauthoryear{Han and Makowski}{2008}]{han2008connectivity}
Han, G., and Makowski, A.~M.
\newblock 2008.
\newblock Connectivity in one-dimensional geometric random graphs: Poisson
  approximations, zero-one laws and phase transitions.

\bibitem[\protect\citeauthoryear{Holland and
  Leinhardt}{1971}]{holland1971transitivity}
Holland, P.~W., and Leinhardt, S.
\newblock 1971.
\newblock Transitivity in structural models of small groups.
\newblock {\em Comparative group studies} 2(2):107--124.

\bibitem[\protect\citeauthoryear{Holland, Laskey, and
  Leinhardt}{1983}]{holland1983stochastic}
Holland, P.~W.; Laskey, K.~B.; and Leinhardt, S.
\newblock 1983.
\newblock Stochastic blockmodels: First steps.
\newblock {\em Social networks} 5(2):109--137.

\bibitem[\protect\citeauthoryear{Maehara}{1990}]{maehara1990intersection}
Maehara, H.
\newblock 1990.
\newblock On the intersection graph of random arcs on the cycle.
\newblock {\em Random Graphs' 87}.

\bibitem[\protect\citeauthoryear{Mossel, Neeman, and
  Sly}{2015}]{mossel2015consistency}
Mossel, E.; Neeman, J.; and Sly, A.
\newblock 2015.
\newblock Consistency thresholds for the planted bisection model.
\newblock In {\em Proceedings of the forty-seventh annual ACM symposium on
  Theory of computing},  69--75.
\newblock ACM.

\bibitem[\protect\citeauthoryear{Paul, Milenkovic, and
  Chen}{2018}]{paul2018higher}
Paul, S.; Milenkovic, O.; and Chen, Y.
\newblock 2018.
\newblock Higher-order spectral clustering under superimposed stochastic block
  model.
\newblock {\em arXiv preprint}.

\bibitem[\protect\citeauthoryear{Penrose and others}{2003}]{penrose2003random}
Penrose, M., et~al.
\newblock 2003.
\newblock {\em Random geometric graphs}, volume~5.
\newblock Oxford university press.

\bibitem[\protect\citeauthoryear{Sahebi and Cohen}{2011}]{pittir13328}
Sahebi, S., and Cohen, W.
\newblock 2011.
\newblock Community-based recommendations: a solution to the cold start
  problem.
\newblock In {\em Workshop on Recommender Systems and the Social Web (RSWEB)}.

\bibitem[\protect\citeauthoryear{Sankararaman and
  Baccelli}{2018}]{sankararaman2018community}
Sankararaman, A., and Baccelli, F.
\newblock 2018.
\newblock Community detection on euclidean random graphs.
\newblock In {\em Proceedings of the Twenty-Ninth Annual ACM-SIAM Symposium on
  Discrete Algorithms},  2181--2200.
\newblock SIAM.

\bibitem[\protect\citeauthoryear{Shi and Malik}{2000}]{shi2000normalized}
Shi, J., and Malik, J.
\newblock 2000.
\newblock Normalized cuts and image segmentation.
\newblock {\em Departmental Papers (CIS)}  107.

\bibitem[\protect\citeauthoryear{Wasserman, Faust, and
  others}{1994}]{wasserman1994social}
Wasserman, S.; Faust, K.; et~al.
\newblock 1994.
\newblock {\em Social network analysis: Methods and applications}, volume~8.
\newblock Cambridge university press.

\bibitem[\protect\citeauthoryear{Yang and Leskovec}{2015}]{yang2015defining}
Yang, J., and Leskovec, J.
\newblock 2015.
\newblock Defining and evaluating network communities based on ground-truth.
\newblock {\em Knowledge and Information Systems} 42(1):181--213.

\bibitem[\protect\citeauthoryear{Zhu, Lafferty, and
  Ghahramani}{2003}]{zhu2003combining}
Zhu, X.; Lafferty, J.; and Ghahramani, Z.
\newblock 2003.
\newblock Combining active learning and semi-supervised learning using gaussian
  fields and harmonic functions.
\newblock In {\em ICML 2003 workshop on the continuum from labeled to unlabeled
  data in machine learning and data mining}, volume~3.

\end{thebibliography}
\normalsize
\bibliographystyle{aaai}

\section*{Supplement}
\section{Main theorem for unsupervised method}
\begin{theorem}[Restating Theorem 4,12 in~\cite{galhotra2018connectivity}]\label{thm:GBM}
    Assume $\theta_1 \geq 2\theta_2$. Define
    \small{
    \begin{align}\label{req:t1234}
        & t_1 = \inf\left\{t\geq 0:(2\theta_2+t)\log(\frac{2\theta_2+t}{2\theta_2})-t>1\right\}\nonumber\\
        & t_2 = \inf\left\{t\geq 0:(2\theta_2-t)\log(\frac{2\theta_2-t}{2\theta_2})+t>1\right\}\nonumber\\
        & t_3 = \sup\Bigg\{t:(4\theta_2+2t_1)\log(\frac{4\theta_2+2t_1}{2\theta_1-t})\nonumber\\
        &+2\theta_1-t-4\theta_2-2t_1>2 \text{ and } 0\leq t\leq 2\theta_1-4\theta_2-2t_1 \Bigg\}\nonumber\\
        & t_4 = \inf\Bigg\{t:(4\theta_2-2t_2)\log(\frac{4\theta_2-2t_2}{2\theta_1-t})\nonumber\\
        &+2\theta_1-t-4\theta_2+2t_2>2 \text{ and }\nonumber\\
        &\max\{2\theta_2,2\theta_1-4\theta_2+2t_2\}\leq t\leq \theta_1 \Bigg\}
    \end{align}
    }
    Then by choosing $E_L = (2\theta_2-t_2)\frac{\log(n)}{n}$ and $E_R = (2\theta_2+t_1)\frac{\log(n)}{n}$, we can recover the correct partition with probability at least $1-o(1)$ if $\theta_1 - t_4 + t_3>2$ or $\theta_1 >\max\{1+t_4,2\}$.
\end{theorem}
In the paper~\cite{galhotra2018connectivity}, their edge removing policy is different from both of our Algorithm~\ref{alg:ALonGBM} and Algorithm~\ref{alg:alter_ALonGBM}. For any edge $(u,v)$, it will be removed if the triangle count fall in the interval $[nE_L,nE_R]$.
\section{Proof of Lemma~\ref{lma:eta_GBM}}
\label{sec:Proof{lma:eta_GBM}}
\begin{lemma_nonum}
    Assume $\theta_2 > 1$. Let
    \scriptsize
    \begin{equation*}
        \eta = \inf\left\{t\geq 0:(\theta_1+\theta_2-2-t)\log(\frac{\theta_1+\theta_2-2-t}{\theta_1+\theta_2-2})+t>1\right\}
    \end{equation*}
    \normalsize
    Then by choosing $E_T = (\theta_1+\theta_2-2-\eta)\frac{\log(n)}{n}$, the underlying communities remains connected in graph $G_r$.
\end{lemma_nonum}

\begin{remark}\label{rmk:eta_detail}
    Note that  we need
    \begin{align*}
        & \eta = \inf\Bigg\{t\geq 0:(\theta_1+\theta_2-2-t)\log(\frac{\theta_1+\theta_2-2-t}{\theta_1+\theta_2-2})+t\\
        & \geq1+\frac{2\log\log(n)}{\log(n)}+\frac{1}{n}\frac{2(2\theta_1-2)t}{\theta_1+\theta_2-2}\Bigg\}.
    \end{align*}
    However, since
    $$
    \frac{2\log\log(n)}{\log(n)}+\frac{1}{n}\frac{2(2\theta_1-2)t}{\theta_1+\theta_2-2} \simeq o(1)
    $$
when $n$ is sufficiently large,  we choose, for simplicity,  $\eta$ as stated in Lemma~\ref{lma:eta_GBM}.  %\Eli{Mazumder et al. do not discuss this point which I think is not rigorous...}
\end{remark}

Before proving Lemma~\ref{lma:eta_GBM}, we need to introduce the following useful lemma
\begin{lemma}\label{lma:totedges}
    Given $\theta_1\geq 1$. With probability at least $1-\frac{1}{n}$, the number of both in-cluster and total edges are $O(n\log(n))$.
\end{lemma}
\begin{proof}
    This has been proved in~\cite{galhotra2018connectivity}. We include the proof here for completeness. Without loss of generality, let $u\in V_1$. For the other arbitrary node $v\in V_1$, the probability of $(u,v)\in E$ is $\frac{2\theta_1\log(n)}{n}$. Thus the number of in-cluster edges associate with $u$ is $Bin(\frac{n}{2}-1,\frac{2\theta_1\log(n)}{n})$. Hence by standard Chernoff bound, for any $c\geq 1$ we have
    \small
    \begin{align*}
        & \mathbb{P}\left(Bin(\frac{n}{2}-1,\frac{2\theta_1\log(n)}{n}) \geq (1+c)(\frac{n}{2}-1)(\frac{2\theta_1\log(n)}{n})\right) \\
        & \leq  \exp\left(-\frac{c}{3}(\frac{n}{2}-1)(\frac{2\theta_1\log(n)}{n})\right) \\
        & = \exp\left(-\frac{c}{3}(1-o(1))\theta_1\log(n))\right)\\
        & = n^{-(1-o(1))\frac{c}{3}\theta_1} \stackrel{(a)}\leq n^{-(1-o(1))\frac{c}{3}},
    \end{align*}
    \normalsize
    where (a) is due to our assumption $\theta_1\geq 1$. Hence if we choose $c>6$, then with probability at least $1-\frac{1}{n^2}$ the number of in-cluster edges associate with $u$ is $O(\log(n))$. By union bound over all nodes, we know that with probability at least $1-\frac{1}{n}$, the number of in-cluster edges is $O(n\log(n))$. For the number of total edges, we can apply the same argument or simply by the fact that $\theta_1 \geq \theta_2$.
\end{proof}
Now we are ready to prove Lemma~\ref{lma:eta_GBM}.
\begin{proof}
From the classical connectivity result of RGG$(\frac{n}{2},\theta\frac{\log(n)}{n})$, we need $\theta > 2$ in order to have $2$ connected components with high probability. Note that the critical radius of $RGG(\frac{n}{2},\theta\frac{\log(n)}{n})$ is: $$\theta\frac{\log(n)}{n} = \frac{\log(\frac{n}{2})}{\frac{n}{2}},$$ which means $\theta = 2-\frac{2\log(2)}{\log(n)} = 2-o(1)$. For simplicity we will set the critical value to be $2$ which is sufficient for the connectivity.

Hence if we will not remove any edges such that $d_{uv}\leq\frac{2\log(n)}{n}$, then the underlying communities remain connected in graph $G_r$.

From Lemma~\ref{lma:Binomials} we know that for an in-cluster edge $(u,v)$ with distance $x = \frac{\phi\log(n)}{n}$, the number of triangles covering such edge is
    \scriptsize
    \begin{align*}
        T^{uv}(x)\sim &Bin(\frac{n}{2}-2,(2\theta_1-\phi)\frac{\log(n)}{n})\\
        &+\mathbf{1}\{\phi\leq 2\theta_2\}Bin(\frac{n}{2},(2\theta_2-\phi)\frac{\log(n)}{n}).
    \end{align*}
    \normalsize
{Note that the mean of both Binomial decreases as $\phi$ increases. Hence for any threshold $t$ independent of $\phi$, we have:
%from which it follows that
%Observe that
% the conditional expectation of $T^{uv}(x)$, given $\sigma(u) = \sigma(v)$, and $d_{uv} = x$,  is monotonically decreasing as $x$ increasing. Hence for any threshold $t$ such that $t<\mathbb{E}\left[T^{uv}(d_{uv})\big|d_{uv}=2\frac{\log(n)}{n},\sigma(u) = \sigma(v)\right]$ and independent of $x$, we have:
\begin{equation*}
    \mathbb{P}\left( T^{uv}(\phi\frac{\log(n)}{n})\geq t\right) \geq \mathbb{P}\left( T^{uv}(2\frac{\log(n)}{n})\geq t\right).
\end{equation*}
for all $\phi\leq 2$ when $\sigma(u) = \sigma(v)$.}
%This is due to the fact that $T^{uv}(x)$ is a Poisson-Binomial random variable \Eli{Need some more verification...}.
Now by assumption we have $2\theta_2 \geq 2$. Then from Lemma~\ref{lma:Binomials} we know that
\begin{align*}
    T^{uv}(\frac{2\log(n)}{n})\sim &Bin(\frac{n}{2}-2,(2\theta_1-2)\frac{\log(n)}{n})\\
    &+Bin(\frac{n}{2},(2\theta_2-2)\frac{\log(n)}{n}).
\end{align*}
Hence it is sufficient to derive the threshold based on the distribution of $T^{uv}(\frac{2\log(n)}{n})$. Next, from Lemma~\ref{lma:totedges} we know that with probability at least $1-\frac{1}{n}$, the number of total edges will be $O(n\log(n))$. Denoting $p = (2\theta_1-2)\frac{\log(n)}{n}$, and $q = (2\theta_2-2)\frac{\log(n)}{n}$,  by Chernoff bound we have:
\scriptsize
\begin{align*}
    &\mathbb{P}\left(T^{uv}(\frac{2\log(n)}{n})\leq nE_T\right) \\
    & \leq \inf_{\xi>0}\exp(\xi nE_T) (1-p+pe^{-\xi})^{\frac{n}{2}-2}(1-q+qe^{-\xi})^{\frac{n}{2}}\\
    &= \inf_{\xi>0}\exp(\xi nE_T+(\frac{n}{2}-2)\log(1+p(e^{-\xi}-1))+\frac{n}{2}(1+q(e^{-\xi}-1)))\\
    &\stackrel{(a)}\leq \inf_{\xi>0}\exp(\xi nE_T+(\frac{n}{2}-2)\log(\exp(p(e^{-\xi}-1)))\\
    &+\frac{n}{2}\log(\exp(q(e^{-\xi}-1))))\\
    &= \inf_{\xi>0}\exp(\xi nE_T+(\frac{n}{2}-2)(p(e^{-\xi}-1))+\frac{n}{2}(q(e^{-\xi}-1)))\\
    &\stackrel{(b)}\leq \exp\Bigg(-\Bigg[(\theta_1+\theta_2-2-\eta)\log(\frac{\theta_1+\theta_2-2-\eta}{\theta_1+\theta_2-2})\\
    &+\eta-\frac{1}{n}\frac{2\eta(2\theta_1-2)}{\theta_1+\theta_2-2}\Bigg]\log(n)\Bigg)\\
    & \stackrel{(c)}\leq \frac{1}{n(\log(n))^2}
\end{align*}
% \exp((-(\theta_1+\theta_2-2-\eta)\log(\frac{\theta_1+\theta_2-2-\eta}{\theta_1+\theta_2-2})-\eta+\frac{4\theta_1-4}{n})\log(n))\\
\normalsize
where for (a) we use $1+x\leq e^x,\forall x\in \mathbb{R}$. For (b) we simply choose $\xi = -\log(\frac{\theta_1+\theta_2-2-\eta}{\theta_1+\theta_2-2})$. Inequality (c) follows by our choice of $\eta$. See Remark~\ref{rmk:eta_detail} for the detail discussion. { Then, by union bound over all edges, we know that,
the probability that  exists at least one edge $(u,v)$  with  $\sigma(u) = \sigma(v)$ and distance $x\leq \frac{2\log(n)}{n}$  such that  its  associated  triangle count satisfies $T^{uv}(x)\leq nE_T$,  is smaller or equal than $O(\frac{n\log(n)}{n(\log(n))^2})$. From this,  it follows immediately that  all the edges $(u,v)$  with  $\sigma(u) = \sigma(v)$ have $T^{uv}(x)\geq nE_T$ with probability
$1-O(\frac{n\log(n)}{n(\log(n))^2})= 1-o(1)$, if  $x\leq \frac{2\log(n)}{n}$. Combining this with the fact that the random geometric graph with $\frac{n}{2}$ nodes will be connected with high probability if the threshold is greater than $\frac{2\log(n)}{n}$, we complete the proof.}
%\Eli{In fact, to be more rigorous we need $\theta_1+\theta_2-(2+c)$ for some small value $c=o(1)$...... If we indeed want to characterize the order of $c$, we need to check the second order phenomenon of the connectivity threshold of the random geometric graph in~\cite{penrose2003random}}
\end{proof}

\section{Proof of Lemma~\ref{lma:crossedgesbound_GBM}}
\begin{lemma_nonum}%\label{lma:crossedgesbound_GBM}
 Assume { $\theta_2 > 1$} and  set $E_T$ as in  Lemma~\ref{lma:eta_GBM}. Let $C$ be the set of cross-cluster edges in $G_r$. If $$\theta_1-\theta_2-2-\eta>t_1,$$ then with probability at least $1-o(1)$ we have $|C| = 0$. If $$t_1>\theta_1-\theta_2-2-\eta>0,$$
    then we have
    \begin{align*}
        |C| \leq \frac{\theta_2}{2}n^{1-\epsilon}(\log(n))^2
    \end{align*}
    with probability at least $1-o(1)$, where
    \begin{equation*}
        \epsilon = (\theta_1+\theta_2-2-\eta)\log(\frac{\theta_1+\theta_2-2-\eta}{2\theta_2})-(\theta_1-\theta_2-2-\eta)
    \end{equation*}
\end{lemma_nonum}

To prove the first half of this lemma, we need the next lemma from~\cite{galhotra2018connectivity}
which provides a bound on the triangle count of any edge
$(u,v)$ such that $\sigma(u)\neq \sigma(v)$.
{Note that when $\sigma(u)\neq\sigma(v)$ the distribution of the triangle count for $(u,v)$ is independent of $d_{uv}$ (see Lemma~\ref{lma:Binomials}). Hence, to simplify the notation, we  just use $T^{uv}$ to denote the triangle count when we condition on $\sigma(u)\neq \sigma(v)$.}

\begin{lemma}[Simplified Lemma 13 in~\cite{galhotra2018connectivity}]\label{lma:simpleLma13}
    Suppose the graph is generated from GBM$(n,\sigma,\frac{\theta_1\log(n)}{n},\frac{\theta_2\log(n)}{n})$. For each edge  $(u,v)$ such that $\sigma(u)\neq \sigma(v)$,  the associated triangle count  $T^{uv}$ is less than $(2\theta_2+t_1)\log(n)$ with probability at least $1-\frac{1}{n\log(n)^2}$.
\end{lemma}
Now we are ready to prove Lemma~\ref{lma:crossedgesbound_GBM}
\begin{proof}
    % \Eli{The original idea is that given $\sigma(u)\neq \sigma(v)$ and $\sigma(u')\neq \sigma(v')$, the triangle counting $T^{uv}$ and $T^{u'v'}$ are only weakly correlated if $(u,v)$ and $(u',v')$ are different edges. Hence we could use Chebyshev inequality to obtain the bound, which is a lot better that Markov inequality. But it turns out that this might not be the case and hence I use only Markov inequality...}

    The proof of the case $\theta_1-\theta_2-2-\eta>t_1$ follows directly from Lemma~\ref{lma:eta_GBM}, Lemma~\ref{lma:simpleLma13}, Lemma~\ref{lma:totedges} and Lemma~\ref{lma:Binomials}. Recall that for any edge $(u,v)$ such that $\sigma(u)\neq \sigma(v)$, $T^{uv}$ is a Binomial random variable with mean $2(1-\frac{2}{n})\theta_2\log(n)$. Note that $\theta_1-\theta_2-2-\eta>t_1$ is equivalent to $\theta_1+\theta_2-2-\eta>2\theta_2+t_1$. By Lemma~\ref{lma:simpleLma13} we know that $T^{uv}$ is less than $(2\theta_2+t_1)\log(n)$ with probability at least $1-\frac{1}{n\log(n)^2}$. This implies that $T^{uv}$ is larger than $nE_T = (\theta_1+\theta_2-2-\eta)\log(n)$ with probability at most $\frac{1}{n\log(n)^2}$. Then by Lemma~\ref{lma:totedges} we know that the number of cross-edges is at most $O(n\log(n))$ with probability at least $1-\frac{1}{n}$. { Applying the union bound over all cross-edges, we show that setting the threshold to $nE_T$ with $E_T$ as in  Lemma~\ref{lma:eta_GBM}, we will remove all cross-cluster edges if $\theta_1-\theta_2-2-\eta>t_1$.}

    Now we prove the result for the case $t_1>\theta_1-\theta_2-2-\eta>0$.
    First we compute the first moment of $|C|$. Let $\mathbf{A}^{(r)}$ be the adjacency matrix of $G_r$. By the edge removing policy in our algorithm, we have
    \begin{align*}
        &\mathbb{E} \left[ |C| \right]= \mathbb{E} \left[ \sum_{u<v,\sigma(u)\neq \sigma(v)}A_{uv}^{(r)} \right] = (\frac{n}{2})^2\mathbb{P}\left(A_{uv}^{(r)}=1\right)
    \end{align*}
    where the last equality is due to symmetry of the cross-edges. Then by Bayes' rule we have
    \begin{align*}
        \mathbb{P}\left(A_{uv}^{(r)}=1\right) &= \mathbb{P}\left(A_{uv}^{(r)}=1\big|A_{uv}=1\right)\mathbb{P}\left(A_{uv}=1\right)\\
        &+\mathbb{P}\left(A_{uv}^{(r)}=1\big|A_{uv}=0\right)\mathbb{P}\left(A_{uv}=0\right)\\
        & = \mathbb{P}\left(A_{uv}^{(r)}=1\big|A_{uv}=1\right)\mathbb{P}\left(A_{uv}=1\right),
    \end{align*}
    where the last equality follows the fact that with probability $1$, $A_{uv} = 0$ implies $A_{uv}^{(r)}=0$ since we do not add edges. It is easy to see that
    \begin{align*}
        \mathbb{P}\left(A_{uv}=1\right) = \frac{2\theta_2\log(n)}{n}
    \end{align*}
    from the generative model. For the other probability, based on Lemma~\ref{lma:Binomials} and our edge removing policy we have
    \begin{align*}
        \mathbb{P}\left(A_{uv}^{(r)}=1\big|A_{uv}=1\right) = \mathbb{P}\left(T^{uv}\leq nE_T\right)
    \end{align*}
    Hence so far we have
    \begin{align}\label{eq:C_moment1_eq}
        \mathbb{E}|C| = \frac{\theta_2n\log(n)}{2}\mathbb{P}\left(T^{uv}\leq nE_T\right)
    \end{align}
    In order to get the more explicit result, we can further upper bound $\mathbb{P}\left(T^{uv}\leq nE_T\right)$ by Chernoff bound as following. For simplicity we denote $q' = \frac{2\theta_2\log(n)}{n}$.
    \begin{align}
    \label{eq:C_moment1_ineq00}
        &\mathbb{P}\left(T^{uv} \leq nE_T\right) \leq \inf_{\xi>0}\exp(\xi nE_T)(1+q'(e^{-\xi}-1))^{n-2}\nonumber \\
        & = \inf_{\xi>0}\exp(\xi nE_T+(n-2)\log(1+q'(e^{-\xi}-1)))\nonumber\\
        & \leq \inf_{\xi>0}\exp(\xi nE_T+(n-2)\log(\exp(q'(e^{-\xi}-1))))\nonumber \\
        & = \inf_{\xi>0}\exp(\xi nE_T+(n-2)q'(e^{-\xi}-1))\nonumber \\
        & \leq \exp\Bigg(-\Bigg[(\theta_1+\theta_2-2-\eta)\log(\frac{\theta_1+\theta_2-2-\eta}{2\theta_2})\nonumber \\
        &-(\theta_1-\theta_2-2-\eta)+\frac{2(\theta_1-\theta_2-2-\eta)}{n}\Bigg]\log(n)\Bigg) \nonumber\\
        & \stackrel{(a)}\leq n^{-\epsilon}
    \end{align}
    Note that inequality (a) is due to our choice of $\epsilon$ and the assumption $(\theta_1-\theta_2-2-\eta)>0$. Using
    \eqref{eq:C_moment1_ineq00} we can upper bound \eqref{eq:C_moment1_eq} as following:
    \begin{align}\label{eq:C_moment1_ineq}
        \mathbb{E} \left[ |C| \right] \leq \frac{\theta_2}{2}n^{1-\epsilon}\log(n)
    \end{align}
    Then, applying  Markov inequality:
    \begin{align}
    \label{ufff}
        \mathbb{P}\left( |C| \geq \mathbb{E} \left[ |C| \right]\log(n)\right) \leq \frac{1}{\log(n)} = o(1)
       \end{align}
    and combining \eqref{ufff} with~\eqref{eq:C_moment1_ineq}, we can show that with probability at least $1-o(1)$, $|C| \leq \frac{\theta_2}{2}n^{1-\epsilon}(\log(n))^2$ which completes the proof.
    % In order to apply Chebyshev inequality, we need further analyze the second moment term of $C$.
    % \begin{align*}
    %     &\mathbb{E}C^2 = \mathbb{E}\left(\sum\limits_{\substack{u<v\\\sigma(u)\neq \sigma(v)}}A_{uv}^{(r)}\right)^2 = \mathbb{E}\left[\sum\limits_{\substack{u<v\\\sigma(u)\neq \sigma(v)}}(A_{uv}^{(r)})^2 + \sum_{u<v,\sigma(u)\neq \sigma(v)}\sum\limits_{\substack{u<v,\sigma(u)\neq \sigma(v)\\ (u,v)\neq (u',v')}}A_{uv}^{(r)}A_{u'v'}^{(r)}\right]
    % \end{align*}
    % For the first term, due to the fact that $(A_{uv}^{(r)})^2 = A_{uv}^{(r)}$ with probability $1$ thus it will be exactly $\mathbb{E}C$. The second term can be further categorized into two cases: $\big|\{u,v\}\cap\{u',v'\}\big| = 1$ or $0$. For the case $\big|\{u,v\}\cap\{u',v'\}\big| = 0$ it can be shown that $$\mathbb{E}A_{uv}^{(r)}A_{u'v'}^{(r)} = \mathbb{E} A_{uv}^{(r)}\mathbb{E}A_{u'v'}^{(r)}.$$
    % Since
    % \begin{align*}
    %     &\mathbb{E}A_{uv}^{(r)}A_{u'v'}^{(r)} = \mathbb{P}\left(A_{uv}^{(r)}=1 \wedge A_{u'v'}^{(r)} = 1\right)\\
    %     &= \mathbb{P}\left(A_{uv}=1 \wedge A_{u'v'} = 1\right)\mathbb{P}\left(A_{uv}^{(r)}=1 \wedge A_{u'v'}^{(r)} = 1\big|A_{uv}=1 \wedge A_{u'v'} = 1\right)
    % \end{align*}
    % Since the nodes $u,v,u',v'$ are all distinct and according to our model, their features are uniformly and independently distributed on a unit length circle. This means $d_{uv}$ is also independent with $d_{u'v'}$ and thus $A_{uv}=1, A_{u'v'}$ are independent. Moreover, the
\end{proof}

\section{Proof of Lemma~\ref{lma:lemmaR}}
\begin{lemma_nonum}%\label{lma:lemmaR}
    Assume $\theta_1 \geq 2\theta_2$, $\theta_2 \geq 1$, and $2\theta_2+t_1 > \theta_1 + \theta_2 - 2 -\eta$. All the in-cluster edges with distance in feature space less than $R$ will not be removed in $G_r$, where
    \scriptsize
    \begin{align}
    \label{mattina}
       R = &\sup_{\min(\theta_1-\theta_2-t_1,2)>r>0} \Bigg\{ (2\theta_2+t_1)\log(\frac{2\theta_2+t_1}{\theta_1+\theta_2-r})\\
       &+(\theta_1+\theta_2-r-(2\theta_2+t_1)) > 1\Bigg\}
            %  & y = (\frac{r}{2}-1)\log(n)-\log(2)
    \end{align}
    \normalsize
\end{lemma_nonum}
\begin{proof}
    From Lemma~\ref{lma:Binomials} we know that for an in-cluster edge $(u,v)$ with distance $x = \frac{\phi\log(n)}{n}$, the number of triangles covering such edge is
    \scriptsize
    \begin{align*}
        T^{uv}(x)\sim &Bin(\frac{n}{2}-2,(2\theta_1-\phi)\frac{\log(n)}{n})\\
        &+\mathbf{1}\{\phi\leq 2\theta_2\}Bin(\frac{n}{2},(2\theta_2-\phi)\frac{\log(n)}{n}).
    \end{align*}
    \normalsize
    Also, the probability that such edge being removed is monotonically {increasing} in $x$ from the similar argument in the proof of Lemma~\ref{lma:eta_GBM}. Then from assumption we know that we can only guarantee to keep all edges with distance $R<2\leq 2\theta_2$ with high probability for some $R$. In order to choose the best possible $R$,
    similarly to Lemma~\ref{lma:eta_GBM},
    we leverage again the tail bound of Poisson Binomial distribution. More specifically, let us denote $p = (2\theta_1-\phi)\frac{\log(n)}{n}$ and $q = (2\theta_2-\phi)\frac{\log(n)}{n}$. The probability that an in-cluster edge $(u,v)$ with $d_{uv} = x = \phi\frac{\log(n)}{n}$ will be removed is
    \scriptsize
    \begin{align*}
        &\mathbb{P}\left(T^{uv}(x) \leq (2\theta_2+t_1)\log(n)\right)\\
        & \leq \inf_{\xi>0}\exp(\xi(2\theta_2+t_1)\log(n)+(\frac{n}{2}-2)\log(1-p+pe^{-\xi}) \\
        &+ \frac{n}{2}\log(1-q+qe^{-\xi}))\\
        & \leq \inf_{\xi>0}\exp(\xi(2\theta_2+t_1)\log(n)+(\frac{n}{2}-2)p(e^{-\xi}-1) + \frac{n}{2}q(e^{-\xi}-1))\\
        & \stackrel{(a)}\leq \exp\bigg(-[(2\theta_2+t_1)\log(\frac{(2\theta_2+t_1)}{\theta_1+\theta_2-\phi})+(2\theta_2+t_1)\\
        &-(\theta_1+\theta_2-\phi)]\log(n)\\
        &+\frac{4\theta_1-2\phi}{\theta_1+\theta_2-\phi}[(\theta_1+\theta_2-\phi)-(2\theta_2+t_1)]\frac{\log(n)}{n}\bigg)\\
        & \stackrel{(b)}\leq \frac{1}{n\log(n)^2}
    \end{align*}
    \normalsize
    Note that for inequality (a), we choose $\xi = -\log(\frac{2\theta_2+t_1}{\theta_1+\theta_2-\phi})$. This is a valid choice for all $0< \phi < \theta_1-\theta_2-t_1$.
    Let $M = \min(\theta_1-\theta_2-t_1,2)$. Now we want to choose the best possible $\phi = R$ for the inequality (b) to hold:
    \scriptsize
    \begin{align*}
        &R = \sup_{M>\phi>0}\Bigg\{(2\theta_2+t_1)\log(\frac{(2\theta_2+t_1)}{\theta_1+\theta_2-\phi})+(2\theta_2+t_1)-(\theta_1+\theta_2-\phi)\\
        &-\frac{1}{n}\frac{4\theta_1-2\phi}{\theta_1+\theta_2-\phi}[(\theta_1+\theta_2-\phi)-(2\theta_2+t_1)]-\frac{2\log(\log(n))}{\log(n)}\geq 1\Bigg\}.
    \end{align*}
    \normalsize
    Note that the last two terms in the above expression are $o(1)$. Hence, for $n$ large enough we can ignore them, from which  \eqref{mattina} follows immediately. Next, note that our assumption $2\theta_2+t_1 > \theta_1+\theta_2-2-\eta$ already implies that $R$ given in \eqref{mattina} is such that $R<2$. See Figure~\ref{fig:explain3} for the illustration. Hence, by using Lemma~\ref{lma:totedges}, which states that that the number of in-cluster edges is $O(n\log(n))$ with probability at least $1-\frac{1}{n}$, by recalling from Section \ref{sec:Proof{lma:eta_GBM}} that the probability  of removing an edge with distance smaller than $R<2$ is at most $\frac{1}{n(\log(n))^2}$, and by applying union bound over all in-cluster edges, it follows immediately that  with probability at least $1-o(1)$ all in-cluster edges with distance less then $R$ will not be removed. This completes the proof.
\end{proof}

\section{Proof of Lemma~\ref{lma:Poissonapprox}}
\begin{lemma_nonum}[Modification of Theorem 8.1 in~\cite{han2008connectivity}]%\label{lma:Poissonapprox}
 Given a random geometric graph, $RGG(n,\tau)$, with $2\tau<1$,  let $\mathcal{C}_{n,\tau}$ be the probability mass  function of the number of  disjoint components $-1$ of $RGG(n,\tau)$ and let $\Pi_\lambda$ denote a Poisson distribution with parameter $\lambda$. Let $d_{TV}(\mu,\nu)\triangleq \frac{1}{2}\sum_{x=0}^{\infty}|\mu(x)-\nu(x)|$ be the total variation of the two  probability mass functions $\mu$ and $\nu$ on $\mathbb{N}$. Then we have
    \begin{align*}
        d_{TV}(\mathcal{C}_{n,\tau},\Pi_{\lambda_n(\tau)}) \leq B_n(\tau),
    \end{align*}
    where
    \scriptsize
    \begin{align*}
        & \lambda_n(\tau) = n(1-\tau)^{n},\;B_n(\tau) = n(1-\tau)^{n} - (n-1)(1-\frac{\tau}{1-\tau})^{n}
    \end{align*}
    \normalsize
\end{lemma_nonum}

\begin{remark}
    This is a slight modification of Theorem 8.1 in~\cite{han2008connectivity}. In fact,  in~\cite{han2008connectivity} the authors assume that the nodes are distributed on an unit length interval rather than on a unit length circle as in our setting. In the following we show how the modification can be done for our case. Note that although we are aware that the authors of~\cite{han2008connectivity} pointed out that the result for the unit length circle had been already established in~\cite{maehara1990intersection}, we have no access to that paper and hence we prove it here again for completeness.
\end{remark}
\begin{proof}
    Let us first assume that we are in the same setting as~\cite{han2008connectivity}, where the nodes are uniformly distributed on the unit length interval. Let $X_1,...,X_n$ denote the ordered node positions, i.e. $0\leq X_1 \leq X_2 \leq ... \leq X_n \leq 1$. Further, let $L_i$ be the spacing, i.e. $L_i = X_i - X_{i-1},\;\forall i=1,...,n+1$ where we let $X_0 = 0$ and $X_{n+1} = 1$. It is clear that $\sum_{i=1}^{n+1}L_i = 1$. Next, let
    $\chi_i(\tau) = \mathbf{1}\{L_i>\tau\},\forall i=2,...,n$. It is clear that $\sum_{i = 2}^{n}\chi_i(\tau)+1$ is exactly the number of disjoint components and thus $\sum_{i = 2}^{n}\chi_i(\tau)$ has probability mass function $\mathcal{C}_{n,\tau}$. Theorem 8.1 in~\cite{han2008connectivity} shows that $$d_{TV}(\mathcal{C}_{n,\tau},\Pi_{\lambda_{n}'(\tau)})\leq B_{n}'(\tau),$$
    \begin{align*}
        & \lambda_n'(\tau) = (n-1)(1-\tau)^{n},\\
        &B_n'(\tau) = (n-1)(1-\tau)^{n} - (n-2)(1-\frac{\tau}{1-\tau})^{n}
    \end{align*}

    Now we turn back to our setting, where the nodes are uniformly distributed on the unit length circle. We observe that by unraveling our circle at an arbitrary place, we obtain the unit length interval and the nodes can be still ordered as $X_1,...,X_n$. Similarly we can define the spacing to be $L_i = X_i - X_{i-1},\;\forall i=1,...,n+1$. However, beside defining $\chi_i(\tau) = \mathbf{1}\{L_i>\tau\},\forall i=2,...,n$, we also define $\chi_1 = \mathbf{1}\{L_1+L_{n+1}>\tau\}$. Note that $L_1+L_{n+1}$ is exactly the distance of $X_1$ to $X_n$ in the circle. Also note that by symmetry, $\chi_1$ is identically distributed as $\chi_2(\tau),...,\chi_n(\tau)$. Thus the number of disjoint components in our case would be $\sum_{i = 1}^{n}\chi_i(\tau)+1$ and hence we sum one more indicator random variable $\chi_1(\tau)$ compare to the case in~\cite{han2008connectivity}. Hence, following the same proof of Theorem 8.1 in~\cite{han2008connectivity} we have Lemma~\ref{lma:Poissonapprox}.
    % \Eli{I could copy all the details in~\cite{han2008connectivity} for completeness if you feel this proof is too rough.}
\end{proof}

\section{Proof of Theorem~\ref{thm:QboundAlgo2}}
\begin{theorem*}%\label{thm:QboundAlgo2}
    Assume  $\theta_1 \geq 2\theta_2 $,  $\theta_2 \geq 1$, and $2\theta_2+t_1 > \theta_1 + \theta_2 - 2 -\eta$. With probability at least $1-o(1)$, Algorithm~\ref{alg:alter_ALonGBM} exactly recovers the clusters with query complexity at most $$\frac{3}{2}n^{1-R/2}+2,$$ where
    \scriptsize
    \begin{align*}
        R = &\sup_{\min(\theta_1-\theta_2-t_1,2)>r>0} \Bigg\{ (2\theta_2+t_1)\log(\frac{2\theta_2+t_1}{\theta_1+\theta_2-r})\\
        &+(\theta_1+\theta_2-r-(2\theta_2+t_1)) > 1\Bigg\}.
    \end{align*}
    \normalsize
\end{theorem*}
%  From the same analysis as in~\cite{galhotra2018connectivity}, we know that with probability at least $1-o(1)$ we remove all cross-cluster edges in \textbf{Phase 1}. Hence in the following, we will focus on bounding the query complexity of our Algorithm~\ref{alg:alter_ALonGBM}.

%  According to our argument above, we know that the query complexity will be the number of disjoint components in $G_r$. Interestingly, one can see that the number of disjoint components of $G_r$ is upper bounded by the number of disjoint components of $2$ $RGG(\frac{n}{2},\frac{R\log(n)}{n})$ for some constant $c$. Let us denote $R'(n) = \frac{2R\log(n)}{\log(n)-\log(2)}$. If $R'(n)\geq 1$ then after the first step we remove all cross-cluster edges in Phase $1$ in Algorithm~\ref{alg:alter_ALonGBM} while we maintain the connectivity within each two clusters. This is due to the connectivity result of RGG. If $R'(n) < 1$ then $RGG(\frac{n}{2},\frac{R\log(n)}{n})$ is no longer connected with high probability and we provide the upper bound for the disjoint component of $RGG(\frac{n}{2},\frac{R\log(n)}{n})$. Together we can establish the query complexity bound for our Algorithm~\ref{alg:alter_ALonGBM}.
\begin{proof}
    From Lemma~\ref{lma:simpleLma13}, we know that with probability at least $1-o(1)$ we remove all cross-cluster edges in \textbf{Phase 1}. Hence in the following, we will focus on bounding the query complexity of our Algorithm~\ref{alg:alter_ALonGBM}.

    Let $Z$ be the number of disjoint components in $G_r$ and let $C_{n/2}(R\frac{\log(n)}{n})+1$ be the number of disjoint components in $RGG(\frac{n}{2},R\frac{\log(n)}{n})$, where $R$ is chosen according to Lemma~\ref{lma:lemmaR}. Note that $C_n(\tau)\sim \mathcal{C}_{n,\tau}$. It is obvious that for all $t$,
    \begin{align}\label{eq:Zbound1}
        \mathbb{P}\left(Z < t\right) \leq \mathbb{P}\left(2(C_{n/2}(R\frac{\log(n)}{n})+1) < t\right)
    \end{align}
    due to our choice of $R$. Then by Lemma~\ref{lma:Poissonapprox} we have
    \begin{equation*}
        d_{TV}(\mathcal{C}_{\tilde{n},\tau},\Pi_{\lambda_{\tilde{n}}(\tau)}) \leq B_{\tilde{n}}(\tau),\;\tilde{n} = \frac{n}{2},\;\tau = R\frac{\log(n)}{n}
    \end{equation*}
    Note that
    \scriptsize
    \begin{align*}
        & B_{\tilde{n}}(\tau) = \tilde{n}(1-\tau)^{\tilde{n}} - (\tilde{n}-1)(1-\frac{\tau}{1-\tau})^{\tilde{n}} \\
        &= \tilde{n}\left[(1-\tau)^{\tilde{n}}-(1-\frac{\tau}{1-\tau})^{\tilde{n}}\right]+(1-\frac{\tau}{1-\tau})^{\tilde{n}}\\
        & \stackrel{(a)}= \tilde{n}(-\tau+\frac{\tau}{1-\tau})\tau^{\tilde{n}-1}\left[1+\frac{1}{1-\tau}+(\frac{1}{1-\tau})^2+\cdots +(\frac{1}{1-\tau})^{\tilde{n}-1}\right]\\
        &+(1-\frac{\tau}{1-\tau})^{\tilde{n}}\\
        & = \tilde{n}\tau^{\tilde{n}+1}\left[\frac{1}{1-\tau}+(\frac{1}{1-\tau})^2+\cdots +(\frac{1}{1-\tau})^{\tilde{n}}\right]+(1-\frac{\tau}{1-\tau})^{\tilde{n}}\\
        & = \tilde{n}\tau^{\tilde{n}+1}\frac{1-(1-\tau)^{-\tilde{n}}}{1-\tau-1}+(1-\frac{\tau}{1-\tau})^{\tilde{n}}\\
        & = \tilde{n}\tau^{\tilde{n}}[(1-\tau)^{-\tilde{n}}-1]+(1-\frac{\tau}{1-\tau})^{\tilde{n}}\\
        & \stackrel{(b)}\leq \tilde{n}\tau^{\tilde{n}}\exp(\tilde{n}\tau)+(1-\frac{\tau}{1-\tau})^{\tilde{n}}
    \end{align*}
    \normalsize
    where for the equality (a) we use the simple algebra fact $x^n-y^n = (x-y)(x^{n-1}+x^{n-2}y+...+y^{n-1})$. The inequality (b) follows from the elementary fact $1+x\leq \exp(x)\;\forall x\in \mathbb{R}$. The first term can be further simplify as
    \scriptsize
    \begin{align*}
        & \tilde{n}\tau^{\tilde{n}}\exp(\tilde{n}\tau) = \exp(\frac{R\log(n)}{2} + \log(\frac{n}{2}) + \frac{n}{2}\log(\frac{R\log(n)}{n}))\\
        & = \exp(\frac{R\log(n)}{2} + \log(\frac{n}{2}) + \frac{n}{2}\log(R\log(n))-\frac{n}{2}\log(n)) = o(1)
    \end{align*}
    \normalsize
    The second term can be bounded as
    \begin{align*}
        & (1-\frac{\tau}{1-\tau})^{\tilde{n}} \leq \exp(-\tilde{n}\frac{\tau}{1-\tau}) = \exp(-\frac{R\log(n)}{2(1-\tau)}) = o(1)
    \end{align*}
    due to the fact $1-\tau = 1-o(1)$. Together we have $B_{\tilde{n}}(\tau) = o(1)$. Thus we have
    \begin{align}\label{eq:Zbound2}
        &\mathbb{P}\left(2(C_{n/2}(R\frac{\log(n)}{n})+1) < t\right) \nonumber \\
        &\leq \mathbb{P}\left(2(\Pi_{\lambda_{\tilde{n}}(R\frac{\log(n)}{n})}+1) < t\right) + o(1)\;\forall t
    \end{align}
    Then we can apply Lemma~\ref{lma:Poisson_tail} to bound the tail probability of Poisson. Choose $y = \frac{1}{2}\lambda$ we have
    \begin{align}\label{eq:Postail1}
        \mathbb{P}\left(\Pi_\lambda \geq \lambda + y \right) \leq \exp(-\frac{y^2}{2(\lambda+y)}) = \exp(-\frac{\lambda}{12})
    \end{align}
    Moreover we have
    \begin{align*}
        \lambda_{\tilde{n}}(R\frac{\log(n)}{n}) = \tilde{n}(1-\frac{R\log(n)}{n})^{\tilde{n}}
    \end{align*}
    It is easy to upper and lower bound this value as following
    \begin{align}
        &\tilde{n}(1-\frac{R\log(n)}{n})^{\tilde{n}} \leq \frac{n}{2}\exp(-\frac{R\log(n)}{2}) = \frac{n^{1-R/2}}{2}\label{eq:Posmean1}\\
        &\tilde{n}(1-\frac{R\log(n)}{n})^{\tilde{n}} \geq \frac{n}{2}\exp(-\tilde{n}\frac{\tau}{1-\tau}) \nonumber\\
        &= \frac{n}{2}\exp(-(1+o(1))\frac{R\log(n)}{2}) = \frac{n^{(1-(1+o(1))R/2)}}{2}\label{eq:Posmean2}
    \end{align}
    where we use the elementary inequality $1-x\geq \exp(1-\frac{x}{1-x})\;\forall x<1$ and the fact that $1-\tau = 1-o(1)$. By our choice of $R$ we know that $1-R/2>0$. Thus plugging  the lower bound~\eqref{eq:Posmean2} into~\eqref{eq:Postail1} we get $\exp(-\frac{\lambda}{12}) = o(1)$ by choosing $\lambda = \lambda_{\tilde{n}}(R\frac{\log(n)}{n})$. This means that with probability at least $1-o(1)$ we have
    \begin{align*}
        \Pi_{\lambda_{\tilde{n}}(R\frac{\log(n)}{n})}\leq \frac{3}{2}\lambda_{\tilde{n}}(R\frac{\log(n)}{n}) \leq \frac{3}{4}n^{1-R/2}
    \end{align*}
Combine this with~\eqref{eq:Zbound1} and ~\eqref{eq:Zbound2} we have shown that with probability at least $1-o(1)$, the number of disjoint components in $G_r$ is at most $$\frac{3}{2}n^{1-R/2}+2,$$
and this completes the proof.
\end{proof}

\section{The $S^2$ algorithm}
Here we include the original $S^2$ algorithm in~\cite{dasarathy2015s2} for completeness. Here we use $f(x)$ to denote the label of node $x$.
\begin{algorithm2e}[htb]
\caption{$S^2$}\label{alg:S2}
\SetAlgoLined
\DontPrintSemicolon
\SetKwInOut{Input}{Input}
\SetKwRepeat{Do}{do}{while}
\SetKwInOut{Output}{Output}
\SetKwInput{Mainalgo}{Main Algorithm}
%\SetKwInput{Phaseone}{Phase 1}
%\SetKwInput{Phasetwo}{Phase 2}
  \Input{Graph $G$, query complexity budget $\mathcal{Q}(\delta)$
  }
  \Output{A partition of $V$
  }
    \Mainalgo{
    $L\leftarrow\emptyset$\;
    \While{1}{
    $x\leftarrow $ Uniformly at random pick an unlabeled node.\;
    \Do{$x\leftarrow MSSP(G,L)$ exists}{
        Add $(x,f(x))$ to $L$\;
        Remove all hyperedges containing nodes with different label from $G$.\;
        \If{more than $\mathcal{Q}(\delta)$ queries are used}{
            \textbf{Return} the remaining connected components of $G$ or \emph{LabelCompletion}$(G,L)$
        }
    }
    }
  }
\end{algorithm2e}

\begin{algorithm2e}[htb]
\caption{MSSP}\label{alg:MMSP}
\SetAlgoLined
\DontPrintSemicolon
\SetKwInOut{Input}{Input}
\SetKwRepeat{Do}{do}{while}
\SetKwInOut{Output}{Output}
\SetKwInput{Mainalgo}{Main Algorithm}
%\SetKwInput{Phaseone}{Phase 1}
%\SetKwInput{Phasetwo}{Phase 2}
  \Input{Graph $G$, label list $L$
  }
  \Output{The midpoint of shortest-shortest path
  }
    \Mainalgo{\;
        \For{each $v,u\in L$ such that $u,v$ has different label}{
        $P_{v,u}\leftarrow $ shortest path between $v,u$ in $G$.\;
        $l_{u,v}\leftarrow $ length of $P_{u,v}$.($=\infty$ if doesn't exist)\;
        }
        $(v^*,u^*) = \arg\min l_{u,v}$\;
        \uIf{$(v^*,u^*)$ exists and $l_{v^*,u^*}\geq2$ }{
            Return the midpoint of $P_{v^*,u^*}$.\;
        }
        \Else{
            Return $\emptyset$\; 
        }
    }
\end{algorithm2e}

Note that the subroutine \emph{LabelCompletion}$(G,L)$ is described in~\cite{dasarathy2015s2} which will perform a certain kind of label propagation algorithm. This is not needed if the query budget is greater then the requirement stated in the theorem of $S^2$. The authors of~\cite{chien2019hs} extend the $S^2$ algorithms to hypergraph setting and more challenging and practical noisy query model.  

\section{Additional experimental results}
\begin{figure}[!h]
  \centering
    % \subfigure[\label{fig:Q_GBM_b3}]{\includegraphics[width=0.32\linewidth]{Q_GBM_b1.eps}}
    \subfloat[]{\includegraphics[width=0.9\linewidth]{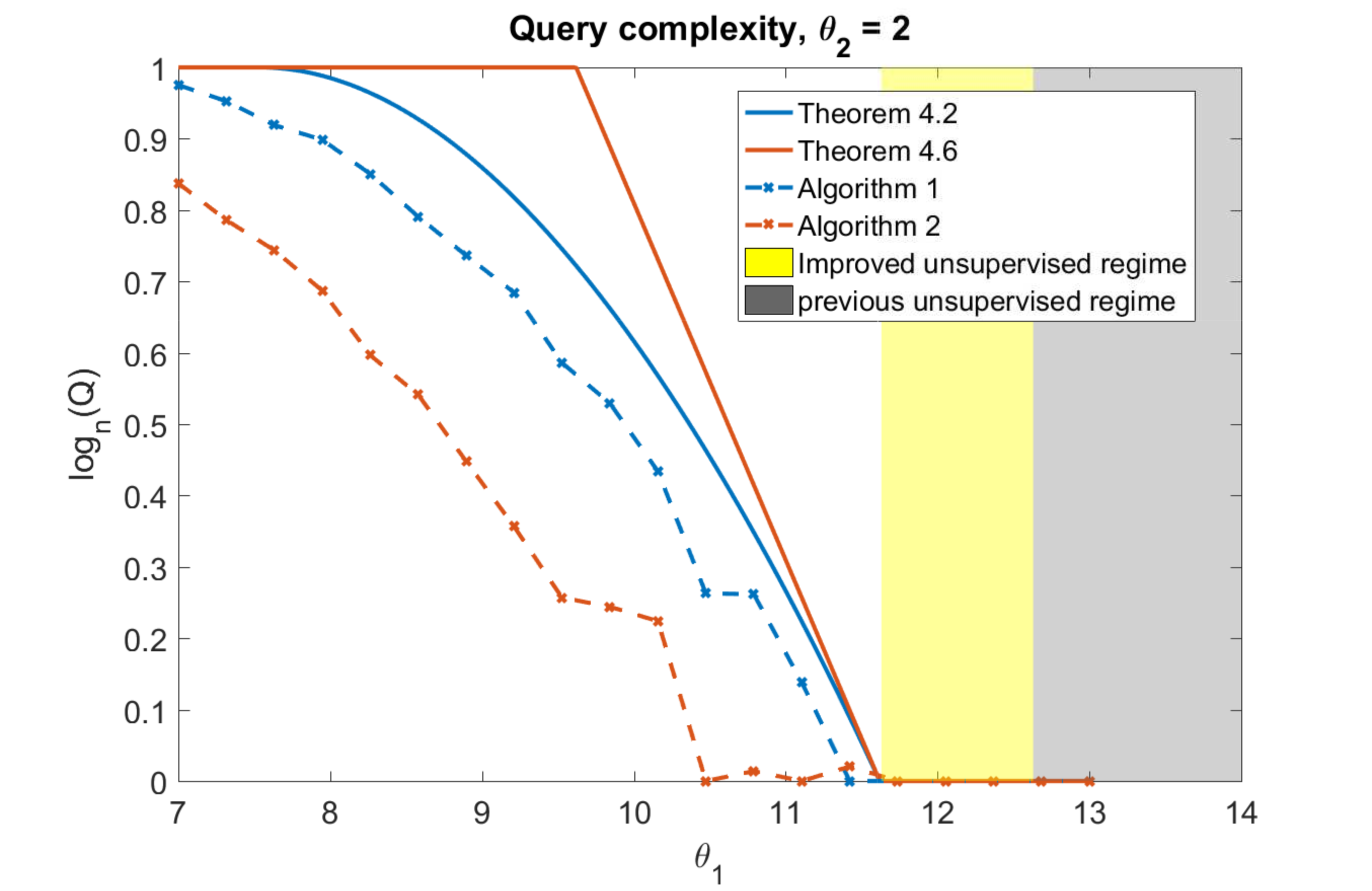}}\\
    \subfloat[]{\includegraphics[width=0.9\linewidth]{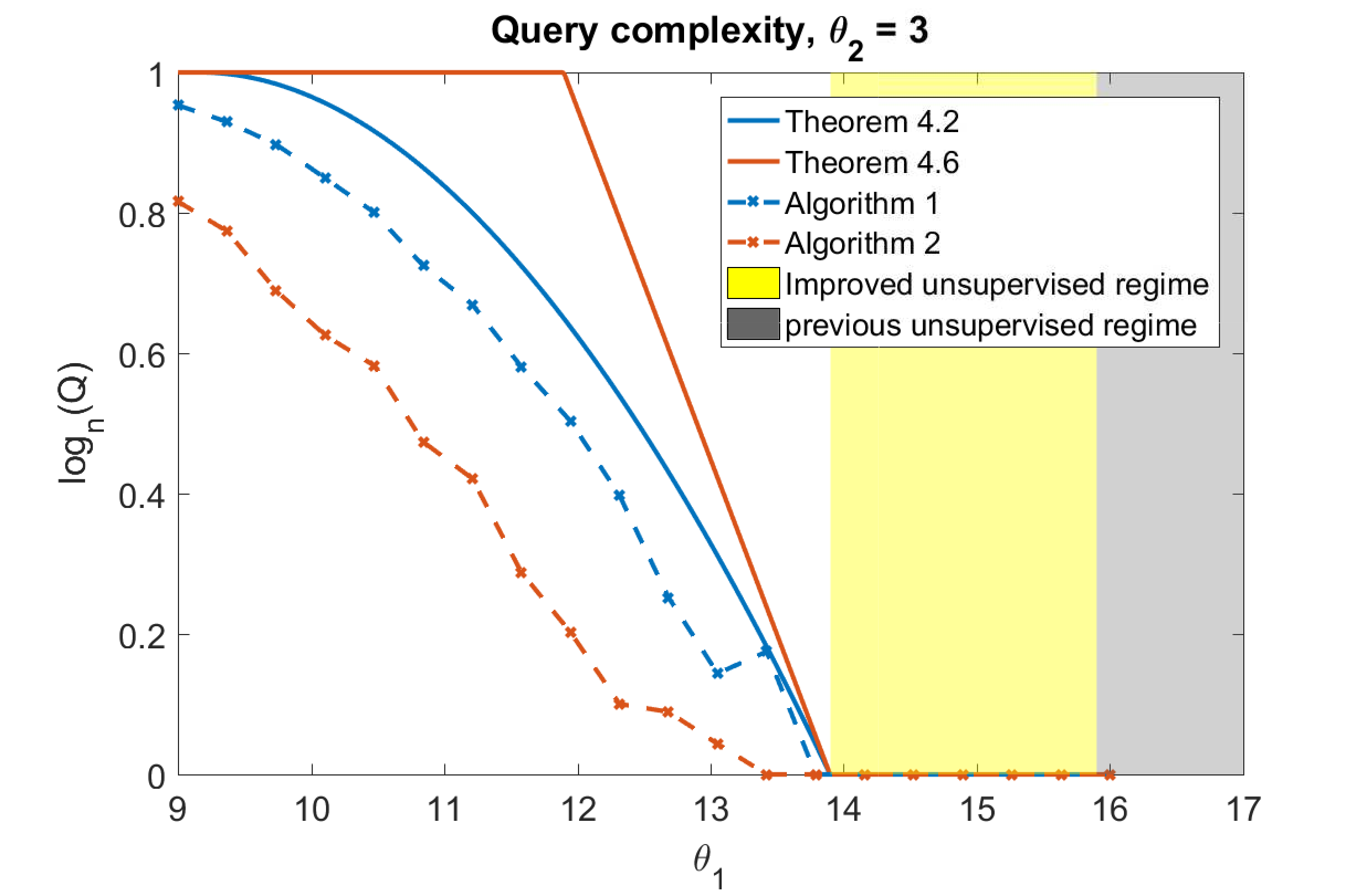}}
  \caption{Additional simulations on the other choice of $\theta_2$.}
\end{figure}

\end{document}